\documentclass[12pt]{article}
\usepackage{amsfonts}
\usepackage{amsmath, amsthm}
\usepackage{epsfig}
\usepackage{algorithm}
\usepackage{algorithmic}
\usepackage{epic}
\usepackage[a4paper]{geometry}
\usepackage{enumitem}
\usepackage{tabularx}
\usepackage{dsfont}
\usepackage{rotating}
\usepackage{lscape}
\usepackage{array}

\usepackage{amsmath,verbatim,color,amssymb,epsfig}
\usepackage{bm}
\usepackage{amsfonts}
\usepackage{epsfig}
\usepackage{changepage}
\usepackage{multirow}
\usepackage{graphicx}
\usepackage{array}
\usepackage[table]{xcolor}
\definecolor{Gray}{gray}{0.80}
\usepackage{lscape}
\usepackage{rotating}
\usepackage[round,semicolon,authoryear]{natbib}
\usepackage[normalem]{ulem}
\usepackage[colorlinks=true,urlcolor=blue,citecolor=purple,linkcolor=blue,bookmarks=true]{hyperref}
\usepackage{caption}
\usepackage{subcaption}

\def\eqx"#1"{{\label{#1}}}
\def\eqn"#1"{{\ref{#1}}}

\makeatletter % make @ act like a letter
\@addtoreset{equation}{section}
\makeatother  % make @ act like a non-letter

\def\squarebox#1{\hbox to #1{\hfill\vbox to #1{\vfill}}}
\def\boxit#1{\vbox{\hrule\hbox{\vrule\kern6pt
          \vbox{\kern6pt#1\kern6pt}\kern6pt\vrule}\hrule}}

\def\theequation{\thesection.\arabic{equation}}

\newtheorem{thm}{Theorem}%[section]

\newtheorem{cor}{Corollary}[section]
\newcolumntype{L}[1]{>{\raggedright\let\newline\\\arraybackslash\hspace{0pt}}m{#1}}
\newcolumntype{C}[1]{>{\centering\let\newline\\\arraybackslash\hspace{0pt}}m{#1}}
\newcolumntype{R}[1]{>{\raggedleft\let\newline\\\arraybackslash\hspace{0pt}}m{#1}}

\usepackage{makecell}
\usepackage{booktabs,multirow}
\usepackage{tikz}
\usetikzlibrary{shapes}

\usepackage{threeparttable}
\usepackage{xurl}

\usepackage{amsthm}

\newtheorem{proposition}{Proposition}

\theoremstyle{remark}

\begin{document}
\thispagestyle{empty}
%% Title
\baselineskip=25pt
\begin{center}
% {\Large \bf  Efficient Treatment Effect Estimation in Longitudinally Conducted Decentralized Clinical Trials}
{\Large \bf On the Calibration of Bayesian Success Criteria and Operating Characteristics for Clinical Trials}
\end{center}
% \vspace{1mm}
%% Author
\begin{center}
{\bf Peng Yang$^{1}$, Li Wang$^{2}$, Ying Yuan$^{1,*}$}
\end{center}

\begin{center}
$^{1}$Department of Biostatistics, The University of Texas MD Anderson Cancer Center, Houston, TX 77030\\
$^{2}$AbbVie Inc., North Chicago, IL, USA\\
{*Email: yyuan@mdanderson.org}
%\vspace{2mm}
\end{center}

\noindent\textbf{Abstract} $\quad$
Recently, the U.S. Food and Drug Administration (FDA) released draft guidance \citep{FDA2026} signaling a paradigm shift that facilitates the use of Bayesian methodology as the primary analysis and decision framework for drug approval. The cornerstone and fundamental challenge of this framework is the specification and calibration of Bayesian success criteria to control decision errors, ensuring reliable clinical and regulatory outcomes. In this work, we systematically investigate various Bayesian decision-error metrics, their theoretical interrelationships, and their alignment with conventional Frequentist counterparts. This investigation provides critical theoretical insights and practical guidance on calibrating Bayesian success criteria and operating characteristics to ensure robust decision-making and the integrity of public health decisions. We illustrate this framework using a clinical trial evaluating revascularization strategies for cardiogenic shock. A Shiny application will be available at \url{www.trialdesign.org} to assist sponsors and regulators in evaluating calibration strategies consistent with recent regulatory perspectives.

%Under this framework, trial success is typically defined by a posterior probability threshold. While one approach is to calibrate this threshold to strict Frequentist Type I error limits, alternative approaches based on direct posterior interpretation are gaining traction \citep{EMA2026}. However, applying these flexible criteria requires a rigorous demonstration of error control to ensure reliable statistical and clinical decisions. To address this critical gap, we systematically investigate various Bayesian decision-error metrics, their theoretical interrelationships, and their alignment with conventional Frequentist counterparts. Through this framework, we elucidate the fundamental trade-offs inherent to strict Type I error calibration, and we further illustrate these dynamics using a clinical case study evaluating revascularization strategies for cardiogenic shock. Ultimately, this work provides theoretical insights and practical guidance to ensure Bayesian clinical trials yield robust conclusions that satisfy rigorous regulatory standards. To facilitate implementation, a Shiny application is available at \url{www.trialdesign.org} to help sponsors and regulators evaluate Bayesian designs consistent with recent regulatory perspectives.

\bigskip


\baselineskip=25pt

\newpage
\setcounter{page}{1} 

\section{Introduction} \label{sec:intro}

Clinical trials are conducted to support regulatory and clinical decision-making regarding the efficacy of new therapeutic interventions \citep{ICH1998E9, pocock1983clinical, FDA1998}. A central element of trial design is the specification of a predefined success criterion, which determines whether the accumulated evidence is sufficient to declare a treatment effective and approve the drug. Historically, the frequentist approach has played a dominant role in confirmatory settings, where statistical significance is typically defined by a $p$-value less than 0.05 for a two-sided test (or 0.025 for a one-sided test). In contrast, Bayesian methods have traditionally been relegated to supplementary analyses.

% In confirmatory settings, study designs are traditionally calibrated using Frequentist operating characteristics, most notably the Type~I error rate and statistical power.

In January 2026, the U.S. Food and Drug Administration (FDA) released draft guidance \citep{FDA2026}, signaling a paradigm shift that facilitates the use of Bayesian methodology as the primary analysis framework for drug approval. As the guidance notes, a common Bayesian success criterion is based on the posterior probability that the treatment effect exceeds a pre-specified, clinically meaningful margin:
\begin{equation} \label{eqn:decision_rule}
\Pr(\theta > \delta \mid \text{data}) > c, \qquad c\in(0,1),
\end{equation}
where $\theta$ represents the true treatment effect, $\delta$ is the minimal clinically meaningful margin established by clinical experts, and $c$ is the posterior probability threshold. This success criterion has been widely adopted in Bayesian clinical trial designs, particularly for monitoring futility, superiority, and safety \citep{thall1994practical, Zhou2017, Liu2015}, as well as for formal regulatory decision-making \citep{Khanna2022}. While alternative metrics are available, most notably Bayesian predictive probability \citep{Lee2008} and the Bayes factor \citep{Kass1995, JohnsonCook2009}, this work focuses on the posterior probability criterion. This choice is motivated by its intuitive interpretation and its central role in the ``evidence-to-decision"  framework favored by modern regulatory guidelines.

The fundamental statistical challenge in implementing the Bayesian (posterior) success criterion (\ref{eqn:decision_rule}) for regulatory and clinical decision-making lies in the determination of the threshold $c$.  The FDA guidance provides two general approaches to calibrate the value of $c$  \citep{FDA2026}. The first approach is to calibrate $c$ to control the Frequentist Type I error rate. In this setting, although the Bayesian decision rule in \eqref{eqn:decision_rule} is utilized, the primary operating characteristics remain conventional Frequentist metrics, such as Type I error rate and power, typically evaluated via simulation. Bayesian methods therefore function primarily as modeling or computational tools within an overall Frequentist inferential framework. To date, this remains the most prevalent approach in confirmatory Bayesian trials.

The second approach leverages Bayesian frameworks to calibrate $c$, drawing on the direct interpretation of posterior probabilities, benefit-risk assessments, or decision-theoretic methods. For instance, setting $c = 0.95$ implies that a trial is deemed successful if there is at least a 95\% posterior probability of effectiveness. While this approach offers significant flexibility, it must not compromise the integrity of clinical and regulatory decision-making. It remains vital to ensure that, for any value of $c$ derived from a specific Bayesian framework, the resulting decision errors, which are not restricted to the Frequentist Type I error rate,  are maintained at levels pre-agreed upon by the sponsor and regulatory authorities. This alignment is a core component of the robust operating characteristics required for regulatory submissions \citep{FDA2026}. 
%Mathematically, calibrating $c$ based on the direct interpretation of posterior probabilities or benefit-risk trade-offs can be reframed as an equivalent calibration to control specific decision errors, such as the Probability of Incorrect Decision (PID), at a pre-specified target level.

%This approach offers significant flexibility, but it 
%%This approach offers significant flexibility, particularly when Frequentist Type I error calibration is problematic—such as when using informative priors to borrow information. 
%%In such cases, enforcing strict Type I error control can inadvertently penalize the very information being incorporated, undermining the efficiency of the design. 
%%However, this flexibility 
%must not compromise the quality of clinical and regulatory decision-making. It remains critical to ensure that, for any selected value of $c$ derived from a specific Bayesian framework, the resulting decision errors are maintained at levels pre-agreed upon by the sponsor and regulatory authorities.  This alignment is a core component of the robust operating characteristics required for clinical trial submissions \citep{FDA2026}. Thus, mathematically, calibrating $c$ based on direct interpretation of posterior probabilities, benefit-risk assessments, or decision-theoretic methods using can be equivalently framed as calibration to control a certain decision errors (e.g., probability of incorrect decision as described later) at a certain level. 

There has been a long-standing debate regarding whether the primary goal of a clinical trial should be statistical inference or binary decision-making \citep{Spiegelhalter1994, Lindley1994}. While one might argue that the role of statistical analysis is to summarize the strength of evidence, fully provided by the posterior distribution, the practical reality of regulation necessitates a binary ``approve or reject" decision. Consequently, the rigorous control of decision errors remains vital. This view has been highlighted in both the 2026 FDA Draft Guidance and the European Medicines Agency (EMA) Concept Paper on the use of Bayesian methods in clinical development \citep{EMA2026, FDA2026}. Both regulatory bodies emphasize that while Bayesian methods offer a rich summary of evidence, their application in confirmatory settings requires a clear demonstration of error control to ensure the integrity of the public health decision. Accordingly, this paper focuses on the decision-making perspective, rather than the inferential summary. For discussion on how Bayesian decision rules, particularly those involving multiple interim looks, impact statistical inference, see \citet{Rosenbaum84} and \citet{liu2025repeated}.

The objective of this paper is to address this critical gap by systematically investigating various Bayesian decision-error metrics, their theoretical interrelationships, and their alignment with conventional Frequentist Type I error calibration. Through this investigation, we aim to elucidate the pros and cons of competing methodologies—specifically, those requiring strict Type I error calibration versus those calibrated to Bayesian decision-error metrics. This work provides both theoretical insights and practical guidance to ensure that Bayesian clinical trial designs yield robust and reliable decisions that meet the rigorous standards of regulatory authorities.

The remainder of this paper is organized as follows. In Section~\ref{sec:method}, we introduce the Bayesian decision rule and formalize the operating characteristics in terms of decision reliability and error rates. In Section~\ref{sec:OC_connection}, we establish analytical connections between these Bayesian operating characteristics and their conventional Frequentist counterparts, characterizing their fundamental relationships. We then evaluate these Bayesian metrics within the context of a clinical case study for cardiogenic shock in Section~\ref{sec:casestudy}. Finally, we conclude with a discussion in Section~\ref{sec:Discussion}.

\section{Methods} \label{sec:method}

\subsection{Success criterion and decision framework}

Consider a superiority trial designed to evaluate the efficacy of an experimental treatment relative to a prespecified reference, informed by historical or external data in a single-arm design or by a concurrent control in a two-arm randomized controlled trial (RCT). For binary endpoints, treatment-effect estimand $\theta$ may correspond to the response probability in a single-arm trial and to the risk difference or ratio between treatment arms in an RCT. For continuous endpoints, $\theta$ may represent the mean outcomes in the single-arm setting or the difference in mean outcomes between arms. For time-to-event endpoints, $\theta$ may represent  the log hazard ratio comparing the experimental treatment with control. Without loss of generality, we assume that larger values of $\theta$ correspond to greater treatment efficacy.

Given the clinical margin $\delta$, the true treatment effectiveness is defined by two complementary states: effectiveness ($\mathcal{E}$) and ineffectiveness ($\bar{\mathcal{E}}$), where:
\[
\mathcal{E}=\{\theta>\delta\}, \qquad 
\bar{\mathcal{E}}=\{\theta\le \delta\}.
\]
%This formulation accommodates both single-arm and RCT designs. 
%In a single-arm trial, the margin $\delta$ is interpreted relative to a clinically meaningful reference value, typically informed by historical trials, external control data, or subject-matter considerations. In contrast, in an RCT, $\delta=0$ corresponds to no treatment difference between arms, and no additional reference parameter is required. 
Let $\mathcal{D}$ denote the observed trial data. We declare trial success if $\Pr(\theta > \delta \mid \mathcal{D}) > c$. Accordingly, we define the decision indicator as follows:
\begin{equation} \label{eqn:decision_indicator}
\mathcal{S}(\mathcal{D}; c)
= \mathbb{I}\!\left\{\Pr(\theta>\delta \mid \mathcal{D})>c\right\},
\end{equation}
where $\mathcal{S}(\mathcal{D};c) = 1$ indicates success and $\mathcal{S}(\mathcal{D};c) = 0$ indicates failure. From a decision-making perspective, such rules naturally induce a $2\times2$ structure that cross-classifies the latent treatment state (effective vs.\ ineffective) against the trial decision (success vs.\ failure), as shown in Table~\ref{tbl:2by2}. Within this framework, a decision may be correct, resulting in a true positive (TP) or true negative (TN), or erroneous, resulting in a false positive (FP) or false negative (FN).

To quantify these decision-error rates under the Bayesian paradigm, where $\theta$ is a random variable,  it is critical to distinguish between the analysis prior $\pi_{\mathrm{a}}(\theta)$ and the design (or sampling) prior $\pi_{\mathrm{d}}(\theta)$ \citep{FDA2026}. The analysis prior represents the distribution assumed by the investigator to conduct posterior inference and calculate $\Pr(\theta>\delta \mid \mathcal{D})$ for decision making. Common choices include noninformative or vague priors. % such as $\text{Beta}(1, 1)$ for binary endpoints or $N(0, 10^6)$ for continuous endpoints. 
In contrast, the design prior represents the distribution that governs the underlying data-generating process, from which $\theta$ is drawn and the subsequent trial data $\mathcal{D}$ are generated. 

While the analysis prior is a known component of the decision rule, the design prior is typically unknown and must be hypothesized during trial planning and operating characteristics assessment. Relative to the clinical margin $\delta$, the design prior induces the prior prevalence of effective and ineffective treatments, denoted by $\gamma_1 = \Pr(\mathcal{E})=\Pr(\theta>\delta)$ and $\gamma_0 = \Pr(\bar{\mathcal{E}})=\Pr(\theta\le \delta)$, respectively. 

It is important to note that in the standard approach of calibrating the threshold $c$ to control the Frequentist Type I error, the design prior becomes irrelevant. This occurs because the Frequentist Type I error is defined as $\alpha(c) = \Pr\{\mathcal{S}(\mathcal{D}; c) = 1 \mid \theta = \delta\}$, which is strictly conditioned on the point null hypothesis $\theta=\delta$. Under this framework, the distribution of $\theta$ is collapsed to a single point ($\delta$), effectively bypassing the influence of the design prior during the calibration process. However, because the decision indicator $\mathcal{S}(\mathcal{D}; c)$ is governed by the posterior probability, $\alpha(c)$  depends on the specification of the analysis prior.

%========================================================
\subsection{Operating characteristics} \label{sec:oc_metrics}
%========================================================
\noindent In this section, we define the key operating characteristic metrics used to evaluate decision reliability and the error risks. While our discussion focuses on Bayesian success criterion derived from the posterior distribution as defined in (\ref{eqn:decision_rule}), these metrics are general and applicable to any probabilistic decision rule.

We frame these metrics within a unified $2 \times 2$ decision framework (Table~\ref{tbl:2by2}), which allows us to quantify not only traditional ``pre-trial" operating characteristics, such as power and Type I error, but also the ``post-trial" reliability of trial conclusions. This dual perspective is essential for both sponsors and regulators to understand the real-world implications of a chosen posterior threshold $c$, particularly regarding the probability that a trial's declaration of success or failure accurately reflects the underlying treatment effect.

\bigskip

\noindent\textbf{Bayesian power (BP)}:  defined as the marginal probability of meeting the success criterion, averaged over the design prior distribution  \citep{spiegelhalter2004bayesian}:
\begin{align} \label{eqn:bayesian_power}
\beta_B(c) \equiv\;& \Pr(\text{success})
=   \int \Pr\{\mathcal{S}({\cal D};c)=1\mid \theta\}\,
\pi_{\mathrm{d}}(\theta)\,d\theta .
\end{align}
BP represents the marginal predictive probability that the trial will declare success across the full range of treatment effects considered plausible \textit{a priori} under the design prior. In the context of the $2 \times 2$ decision framework (Table~\ref{tbl:2by2}), BP corresponds to the total proportion of successful outcomes: $(\mathrm{TP} + \mathrm{FP}) / (\mathrm{TP+TN+FP+FN})$. When the design prior is a point mass at the alternative value of $\theta$, Bayesian power $\beta_B(c)$ coincides with the conventional Frequentist power.

A notable limitation of BP is that it quantifies only the propensity of the trial to declare success, rather than the reliability of that decision. This is because the integral in \eqref{eqn:bayesian_power} averages over the entire support of the design prior, encompassing both the effectiveness region ($\theta > \delta$), where success constitutes a correct decision, and the ineffectiveness region ($\theta \le \delta$), where success represents a false positive error. To address this, one may decompose the marginal probability by integrating separately over these two regions. This refinement leads to the definitions of Bayesian conditional power and the Bayesian Type I error rate, as detailed below.

\noindent\textbf{Bayesian conditional power (BCP)}: defined as the probability of meeting the success criterion given that the treatment is effective, where the expectation is taken with respect to the design prior restricted to the effectiveness region:  
\begin{align*}
\beta_{C}(c)
\equiv &
\Pr(\text{success}\mid \text{effective}) 
=  \int_{\theta > \delta } \Pr\{\mathcal{S}({\cal D};c)=1\mid \theta\}\,
\pi_{\mathrm{d}}(\theta \mid \theta > \delta)\,d\theta .
\end{align*}
BCP represents the probability that the trial declares success when the treatment is truly effective, as characterized by the design prior. In the $2 \times 2$ decision framework (Table~\ref{tbl:2by2}), BCP corresponds to the True Positive Rate: $\mathrm{TP} / (\mathrm{TP} + \mathrm{FN})$. BCP serves as a Bayesian counterpart to Frequentist power; however, while Frequentist power is typically conditioned on a single point-mass alternative (e.g., $\theta = \theta_A$, where $\theta_A > \delta$), BCP incorporates the uncertainty of the treatment effect by averaging the success probability over the entire effectiveness region defined by the design prior. 
 
\noindent\textbf{Bayesian Type I error rate:} defined as the probability that the trial incorrectly declares success, conditioned on the treatment being ineffective. The expectation is taken with respect to the design prior restricted to the ineffectiveness region:
\begin{align*}
\alpha_B(c)
\equiv
& \Pr(\text{success}\mid \text{ineffective}) = \int_{\theta \leq \delta } \Pr\{\mathcal{S}({\cal D};c)=1\mid \theta\}\,
\pi_{\mathrm{d}}(\theta \mid \theta \leq \delta)\,d\theta .
\end{align*}
The Bayesian Type I error rate represents the average risk of a false positive finding across all plausible ineffective treatment states \citep{Best2025}.
In the $2 \times 2$ decision framework (Table~\ref{tbl:2by2}), it corresponds to the False Positive Rate: $\mathrm{FP} / (\mathrm{FP} + \mathrm{TN})$. This metric serves as a Bayesian counterpart to the Frequentist Type I error rate; however, while the Frequentist approach typically conditions on a single point-mass null (e.g., $\theta = \delta$), the Bayesian Type I error rate incorporates uncertainty regarding the null state by averaging the success probability over the entire ineffectiveness region defined by the design prior. This metric is recognized in 2010 FDA guidance for the use of Bayesian statistics in medical device clinical trials as an alternative definition of Type I error for Bayesian methods \citep{FDA2010}. Furthermore, the January 2026 FDA Guidance on Bayesian Methodology recommend that ``design priors that explore pessimistic assumptions about treatment effects should also be considered" when evaluating Bayesian power, which under pessimistic scenarios is closely related to the Bayesian Type I error rate.

%--------------------------------------------------------
\noindent\textbf{Probability of incorrect decision (PID)}: defined as the probability that the treatment is ineffective given that the trial has declared success:
\begin{equation*}
\mathrm{PID}(c)
\equiv
\Pr(\text{ineffective} \mid \text{success}) = \Pr\!\left(\theta \le \delta \;\middle|\; \mathcal{S}({\cal D};c)=1\right).
\end{equation*}
The PID addresses a critical question for decision-makers: Given that the trial has met the success criterion, what is the probability that the treatment is actually ineffective? In the context of the $2 \times 2$ decision framework (Table~\ref{tbl:2by2}), the PID is  $\mathrm{FP}/(\mathrm{TP}+\mathrm{FP})$.

As emphasized in the 2026 FDA guidance \citep{FDA2026}, the PID differs fundamentally from Frequentist or Bayesian Type I error; while the latter is a pre-data measure conditioned on the treatment being ineffective ($\bar{\mathcal{E}}$), the PID is a post-data (or predictive) measure conditioned on the observed success event ($\mathcal{S}=1$). Consequently, the PID aligns more naturally with the Bayesian paradigm of conditioning on observed data, making it a useful metric for quantifying regulatory decision error. If a successful trial leads to market approval, the PID represents the probability that the approved drug is truly ineffective. Ultimately, the PID ensures that ``if we declare a drug successful, we are likely correct", whereas the Type I error ensures ``if a drug is ineffective, we are unlikely to declare it successful.'"

Mathematically, PID is equivalent to the complement of the Positive Predictive Value (PPV; \citet{Altman1994}) in diagnostic testing ($1 - \text{PPV}$) or the positive False Discovery Rate (pFDR) \citep{storey2003positive}. Similar to PPV \citep{Wacholder2004}, a defining characteristic of the PID is its intrinsic sensitivity to the design prior, specifically, the prior prevalence of effectiveness $\gamma_1$, as follows:
\begin{equation*} \label{eqn:PID_formula}
\mathrm{PID}(c) = \frac{\alpha_B(c) \gamma_0}{\beta_C(c) \gamma_1 + \alpha_B(c) \gamma_0}.
\end{equation*}

Consequently, $\mathrm{PID}(c)$ may be significantly larger or smaller than the Bayesian Type I error rate $\alpha_B(c)$ under the same cutoff $c$. While the Frequentist Type I error rate is invariant to the prevalence of effective treatments and the Bayesian Type I error rate is normalized to the ineffective region $\bar{\cal E}$, the PID measures a fundamentally distinct dimension of decision reliability. This relationship is directly analogous to the distinction between PPV and sensitivity in diagnostic screening,
% whereas Type I error rate (sensitivity's counterpart) quantifies the risk of failing the null, PID quantifies the ``post-test" probability that a positive finding is a false discovery. This 
highlighting the necessity of accounting for the design prior when evaluating the real-world reliability of positive trial results, as further elaborated in section \ref{sec:PID}. 

%--------------------------------------------------------
\noindent\textbf{False omission rate (FOR)}: defined as the probability that the treatment is effective among trials that fail to declare success:
\begin{equation*}
\mathrm{FOR}(c)
\equiv  \Pr(\text{effective}\mid \text{failure}) = \Pr\!\left(\theta>\delta \;\middle|\; \mathcal{S}({\cal D};c)=0\right).
\end{equation*}
In the $2 \times 2$ decision framework (Table~\ref{tbl:2by2}), the FOR corresponds to $\mathrm{FN}/(\mathrm{FN}+\mathrm{TN})$. This quantifies the risk of missing a beneficial treatment under the chosen decision rule and complements $\mathrm{PID}(c)$, which focuses on reliability among declared successes. In the regulatory setting, the FOR represents the percentage of drugs that fail to obtain approval despite being truly effective. Clearly, a trade-off exists between PID and FOR: setting a stringent success criterion (i.e., a larger $c$) reduces the PID but increases the FOR. An excessively high FOR diminishes the overall societal benefit of drug development by incorrectly rejecting effective therapies, thereby depriving patients of potentially life-saving treatments. 

The FOR is fundamentally linked to the diagnostic concepts of sensitivity and specificity through the design prior (see Section~\ref{sec:PID}). Mathematically, the FOR is the complement of the Negative Predictive Value (NPV) and is equivalent to the False Undiscovery Rate (FUnR) frequently utilized in multiple testing and genomics. Unlike specificity, which is invariant to the prevalence of effective treatments, the FOR is highly dependent on the design prior. Specifically, as the prior mass assigned to the effective region $\mathcal{E}$ increases, representing a higher ``pre-trial'' probability of drug efficacy, the FOR  tends to rise, even if the specificity of the design remains constant. This reinforces the importance of the design prior in evaluating the reliability of negative trial results, as it captures the post-trial probability of an erroneous negative conclusion.

For continuous and binary endpoints, the aforementioned operating characteristic metrics admit closed-form expressions, which are summarized in Table~\ref{tbl:OC_summary} and Table~\ref{tbl:OC_binary_summary} in the Supplementary Materials, respectively. 

% To illustrate their behavior, Figure~\ref{fig:oc_normal_design_prior_sd_0.15} presents these operating characteristics under a single-arm continuous-outcome setting with various normal design priors. 

%Crucially, these expressions remain broadly applicable to non-normal outcomes, such as binary or time-to-event data, because the distributions of their associated test statistics can generally be well approximated by a normal distribution via the Central Limit Theorem. Detailed derivations are provided in the Appendix. Consequently, these analytical solutions establish a robust and computationally efficient framework for evaluating decision reliability across a wide range of clinical trial designs.

% For normal endpoints, the aforementioned operating characteristic metrics admit closed-form expressions, which are summarized in Table~\ref{tbl:OC_summary}. Crucially, these expressions remain broadly applicable to non-normal outcomes, such as binary or time-to-event data, because the distributions of their associated test statistics can generally be well-approximated by a normal distribution via the Central Limit Theorem. Detailed derivations are deferred to the Appendix. Consequently, these analytical solutions establish a robust and computationally efficient framework for evaluating decision reliability across a wide variety of clinical trial designs.

\section{Relationship between operating characteristic metrics} \label{sec:OC_connection}

In this section, we examine the mathematical and practical relationships between the operating characteristic metrics defined previously. These results delineate the connection between the common Frequentist-inspired approach of calibrating the success threshold $c$ via Type I error and the Bayesian approach of selecting $c$ based on decision-reliability metrics, specifically the PID. The latter approach may appear distinct from methods based on the direct interpretation of posterior probabilities, benefit-risk assessments, or decision-theoretic utilities; however, it effectively encompasses these frameworks through the following practical consideration. When adopting a subjective derivation of $c$—such as those rooted in clinical utility or the direct interpretation of posterior probabilities—sponsors are almost invariably required to demonstrate that the resulting threshold maintains decision errors within pre-specified nominal bounds. This is a core component of the robust operating characteristics mandated for regulatory submissions \citep{FDA2026}. Since the success criterion in (\ref{eqn:decision_rule}) relies on a single tuning parameter, $c$, the process of calibrating this threshold (whether via benefit-risk trade-offs or direct probability interpretation) while maintaining error control can be straightforwardly reframed as an equivalent calibration designed to bound a specific decision error, such as the PID, at a target level.
%As noted in the introduction, while current FDA guidance \citep{FDA2026} supports selecting $c$ through a direct interpretation of the posterior distribution, researchers still face the challenge of determining a specific threshold  value. Whether $c$ should be set at $0.80$, $0.90$, or $0.95$ remains a critical decision point. While this choice is informed by clinical characteristics—such as the target patient population and the effectiveness of the standard of care—it must ultimately be anchored to one or more formal operating characteristic criteria. Indeed, the FDA guidance underscores that regardless of the method used to choose $c$, the resulting design must demonstrate desirable operating characteristics, such as reasonable PID.

By investigating the interplay between these operating characteristic metrics, we provide evidence-based answers to common questions encountered in Bayesian clinical trial design: Should the design be calibrated using Type I error or alternative Bayesian metrics? If the latter is chosen, what criteria are most appropriate, and how do they relate to traditional Frequentist benchmarks? Our analysis offers a framework for navigating these trade-offs to ensure robust and reliable trial conclusions.

\subsection{Bayesian power and Type I error}
The following results establish the relationship between Bayesian Power ($\beta_B(c)$), Bayesian Conditional Power ($\beta_C(c)$), the Bayesian Type I Error rate ($\alpha_B(c)$), and the Frequentist Type I Error rate ($\alpha(c)$).

\begin{proposition} \label{prop:BP_decomp_CP}
The Bayesian power $\beta_B(c)$ admits the decomposition
\begin{align*}
\beta_B(c) = \gamma_1\,\beta_C(c)+\gamma_0\,\alpha_B(c),
\label{eqn:BP_decomp_CP}
\end{align*}
where $\gamma_1=\Pr(\theta>\delta)$ and $\gamma_0=1-\gamma_1$. In practical settings where $\alpha_B(c) < \beta_C(c)$, it follows that $\beta_B(c) \le \beta_C(c)$. Furthermore, if the design prior $\pi_{\mathrm{d}}(\theta)$ concentrates entirely on the effective region ($\gamma_1 \to 1$), then $\beta_B(c)\to \beta_C(c)$.
\end{proposition}

\begin{proposition} \label{prop:FPR_le_T1E}
The Bayesian Type I error rate is bounded by the Frequentist Type I error rate:
\begin{equation*}
\alpha_B(c) \le \alpha(c).
\end{equation*}
Equality holds if and only if the design prior $\pi_{\mathrm{d}}(\theta)$ is a point mass at the boundary point $\theta = \delta$.
\end{proposition}

\begin{proposition}
\label{prop:alpha_asymp}
Under  standard regularity conditions \citep{van2000asymptotic}, 
\[
\alpha(c) %=\Pr\{\mathcal{S}(\mathcal{D};c)=1 \mid \theta=\delta\}
\to 1-c
\quad \text{as } n\to\infty.
\]
\end{proposition}

Proposition~\ref{prop:BP_decomp_CP} demonstrates that in most practical settings, where trials are designed with high power and low Type I error, calibrating to Bayesian power is a more stringent requirement than calibrating to Bayesian conditional power. Consequently, for a fixed value of $c$, the Bayesian power will be lower than the conditional power (see Figure~\ref{fig:oc_normal_design_prior_sd_0.15}). 

As previously noted, calibrating solely to Bayesian power may not be advisable, as it reflects only a marginal probability of success that conflates true positives and false positives. While informative to a sponsor for resource planning, it is less informative for regulatory bodies primarily concerned with false positive risks. 
Consequently, we recommend that the appropriateness of the threshold $c$ (and the corresponding sample size) be evaluated primarily through Bayesian conditional power, which directly reflects the reliability of a successful trial outcome. Bayesian power should be monitored as a secondary feasibility metric.  As illustrated in Figure~\ref{fig:oc_normal_design_prior_sd_0.15}, the Bayesian power converges toward the Bayesian conditional power as the prevalence of effective treatments $\gamma_1$ increases.

Proposition~\ref{prop:FPR_le_T1E} indicates that the Frequentist Type I error $\alpha(c)$ controls false positive risk in a ``worst-case" sense, as it is defined strictly at the null hypothesis boundary. Therefore, given the same nominal significance level, a design calibrated to control the Bayesian Type I error will result in a more relaxed success criterion (a smaller value of $c$) compared to one calibrated to Frequentist Type I error (see Figure~\ref{fig:oc_normal_design_prior_sd_0.15}). 

Proposition~\ref{prop:alpha_asymp} suggests that when calibrating the Bayesian posterior success criterion \eqref{eqn:decision_rule} to a nominal Frequentist Type I error rate, $\alpha^\ast$ (e.g., 0.025), the decision threshold $c$ should be approximately $1 - \alpha^\ast$ (e.g., 0.975), as it asymptotically controls the Frequentist Type I error rate at the level of $\alpha^\ast$. This relationship serves as a robust initial estimate for calibrating $c$ in finite-sample settings when the objective is to maintain Frequentist Type I error control, particularly when a non-informative or weakly informative analysis prior is employed.

A fundamental distinction between these calibration strategies is that $\alpha(c)$ is invariant to the prevalence of effective treatments and pre-trial beliefs. In contrast, $\alpha_B(c)$ is dependent on the design prior $\pi_{\mathrm{d}}(\theta)$, as it averages the false positive risk over the entire ineffective region $\{\theta \le \delta\}$ according to $\pi_{\mathrm{d}}(\theta)$. Specifically, $\alpha_B(c)$ is influenced by: (i) the prevalence of ineffective treatments encoded by $\pi_{\mathrm{d}}(\theta)$, and (ii) the concentration of ineffective treatment effects near the null boundary $\theta = \delta$. Consequently, $\alpha_B(c)$ can be substantially smaller than $\alpha(c)$ when the design prior assigns little mass near the boundary, though it approaches $\alpha(c)$ as the prior mass concentrates close to $\delta$. As shown in Figure~\ref{fig:oc_normal_design_prior_sd_0.15}, the Bayesian Type I error rate remains consistently lower than the Frequentist Type I error rate across a wide range of effective-treatment prevalences $\gamma_1$.

\subsection{PID and FOR}\label{sec:PID}
The following results establish the relationship between PID($c$), FOR($c$), and the Frequentist Type I error rate $\alpha(c)$.  These connections illuminate the link between designs calibrated to the Frequentist Type I error rate and those guided by Bayesian decision-reliability metrics.

\begin{thm}
\label{thm:PID_bound}
Assume the analysis prior $\pi_{\mathrm{a}}(\theta)$ and the design prior $\pi_{\mathrm{d}}(\theta)$ share the same conditional distributions within the effectiveness ($\mathcal{E}$) and ineffectiveness ($\bar{\mathcal{E}}$) regions, i.e.,  $\pi_{\mathrm{a}}(\theta \mid \mathcal{E}) = \pi_{\mathrm{d}}(\theta \mid \mathcal{E})$ and $\pi_{\mathrm{a}}(\theta \mid \bar{\mathcal{E}}) = \pi_{\mathrm{d}}(\theta \mid \bar{\mathcal{E}})$. Then, the Probability of Incorrect Decision $(\mathrm{PID})$ is bounded by:
\begin{equation} \label{eqn:PID_general_bound}
\mathrm{PID}(c) \le \frac{1}{R \left( \frac{\displaystyle c}{\displaystyle 1-c} \right) + 1},
\end{equation}
where $R = {OE}_{\mathrm{d}} / {OE}_{\mathrm{a}}$ is the ratio of the prior odds of effectiveness (OE) under the design prior (i.e., ${OE}_{\mathrm{d}} = \gamma_1/\gamma_0$) to those under the analysis prior (i.e., ${OE}_{\mathrm{a}}$). 
\end{thm}

\begin{cor} \label{cor:PID_bound}
When $R=1$, which is true under $\pi_{\mathrm{a}}(\theta) = \pi_{\mathrm{d}}(\theta)$, the PID is strictly bounded by 
%\begin{equation*} \label{eqn:PID_matched_bound}
$\mathrm{PID}(c) \le 1 - c$.
%\end{equation*}
\end{cor}

Theorem \ref{thm:PID_bound} establishes that the post-trial risk of a false discovery is not a static property of the decision threshold $c$; rather, it is a dynamic value governed by the prior odds ratio $R$, as shown in Equation (\ref{eqn:PID_general_bound}). This result is significant because it suggests that the conventional approach of setting $c \approx 1 - \alpha^*$, which is widely used to control the Frequentist Type I error rate at a nominal level $\alpha^*$ (Proposition~\ref{prop:alpha_asymp}), can also effectively control the PID at that same nominal level, provided the analysis and design priors are perfectly aligned ($R=1$). This alignment provides a unified justification for threshold selection, ensuring that the pre-trial procedural risk (calibrated by Frequentist  Type I error) and the post-trial decision risk (calibrated by PID) are numerically harmonized.

In practice, however, the assumption of matched analysis and design priors may not hold. Frequently, a neutral or non-informative prior with ${OE}_{\text{a}}=1$ is utilized for the analysis, whereas the true data-generating design prior does not necessarily favor effectiveness and ineffectiveness equally. In this common scenario of prior misalignment, simply setting $c=1-\alpha^*$ fails to guarantee control of the PID at the $\alpha^*$ level, as formally indicated by Theorem \ref{thm:PID_bound}. Specifically, when employing a neutral analysis prior (${OE}_{\text{a}} \approx 1$) in high-risk therapeutic areas where prior success rates are low (${OE}_{\text{d}} < 1$), the PID may inflate beyond the nominal level $\alpha^*$ because $R < 1$ (see the Pessimistic scenario in Figure \ref{fig:oc_normal_design_prior_sd_0.15}). Conversely, if an investigator employs an overly skeptical analysis prior in a promising disease area ($R > 1$), the PID  will fall  below the nominal level $\alpha^*$ (Optimistic scenario, Figure \ref{fig:oc_normal_design_prior_sd_0.15}).

Consequently, Equation (\ref{eqn:PID_general_bound}) provides as a dynamic calibration tool: to ensure the realized decision risk remains bounded by a target level $\tau^*$ (e.g., 0.025), the threshold $c$ must be rigorously calibrated against the underlying clinical and prior reality represented by $R$. Because the true design prior $\pi_{\mathrm{d}}(\theta)$ is generally unknown in practice, a comprehensive sensitivity analysis under various specifications of $\pi_{\mathrm{d}}(\theta)$ is essential. As suggested by our results, a natural starting point for this calibration and sensitivity analysis is to set $c = 1 - \tau^*$. We utilize the notation $\tau^*$ as the target nominal level for the PID to emphasize that this threshold is distinct from—and need not be equal to—the target nominal level for Frequentist (or Bayesian) Type I errors.

The theoretical bound established in Theorem \ref{thm:PID_bound} relies on the assumption of congruence between the analysis and design priors within each hypothesis space—specifically, that $\pi_{\mathrm{a}}(\theta \mid \mathcal{E}) = \pi_{\mathrm{d}}(\theta \mid \mathcal{E})$ and $\pi_{\mathrm{a}}(\theta \mid \bar{\mathcal{E}}) = \pi_{\mathrm{d}}(\theta \mid \bar{\mathcal{E}})$. In more general settings where $\pi_{\mathrm{a}}$ and $\pi_{\mathrm{d}}$ vary freely, establishing a closed-form analytic bound becomes significantly more challenging. In such cases, simulation studies are required to assess the range and behavior of the $\mathrm{PID}(c)$. Nevertheless, Theorem \ref{thm:PID_bound} serves as a critical benchmark and a logical starting point for such simulations and the associated sensitivity analyses described above. 

%This dynamic requirement creates an immediate tension with standard Frequentist operating characteristics. In regulatory settings, decision rules are traditionally benchmarked against a fixed Frequentist Type I error rate ($\alpha$). A natural follow-up question arises: when considering the $\mathrm{PID}$ alongside the Type I error rate, what is the exact relationship between the two, given that the analysis prior and the design prior are rarely identical in clinical practice? Specifically, under what underlying clinical conditions and trial operating characteristics does the true Bayesian risk actually breach the accepted Frequentist limit? The answer is provided by the following result.

We have demonstrated that calibrating Bayesian success criteria to control the PID is different from calibrating to control the Frequentist Type I error rate $\alpha(c)$. These two approaches may yield different decision thresholds $c$, even when they share the same target nominal level. The following theorem further establishes the mathematical relationship between these two metrics: 

\begin{thm}
\label{thm:PID_T1E}
% Given a fixed $c$, a necessary and sufficient condition for $\mathrm{PID}(c) < \alpha(c)$ is given by:
Under large-sample settings, given a fixed posterior cutoff $c$, %the condition for $\mathrm{PID}(c) < \alpha(c)$ is approximately given by
$\mathrm{PID}(c) < \alpha(c)$ if and only if
\begin{equation*}
\frac{\beta_C(c)}{\alpha_B(c)} > \frac{\gamma_0}{\gamma_1} \cdot \frac{c}{1-c}.
\label{eqn:PID_T1E}
\end{equation*}
\end{thm}
%This result relies on a large-sample framework where the likelihood of the data dominates the posterior distribution \citep{van2000asymptotic}. Under such settings, particularly when employing a non-informative analysis prior, the Frequentist Type I error rate associated with the posterior success rule approximately satisfies $1-c$. 
This result demonstrates that, for a given posterior probability threshold $c$, there is no strict ordering between $\mathrm{PID}(c)$ and $\alpha(c)$. Whether the PID is smaller or larger than the Type I error rate depends on the interaction of four components: the Bayesian conditional power ($\beta_{C}(c)$); the Bayesian Type I error rate ($\alpha_{B}(c)$), the prior odds of ineffectiveness ($\gamma_0/\gamma_1$) induced by the design prior $\pi_{\mathrm{d}}(\theta)$; and the stringency of the decision rule as defined by the threshold $c$. In general, the Frequentist Type I error $\alpha(c)$ is more likely to exceed the $\mathrm{PID}(c)$, meaning the Frequentist metric is more conservative, when the Bayesian conditional power is high, the Bayesian Type I error rate is low, the prior probability of effectiveness ($\gamma_1$) is high, and the posterior cutoff $c$ is relatively modest.

To illustrate, consider a design where $\beta_C(c) = 0.80$ and $\alpha_B(c) = 0.02$ (recalling that $\alpha_B(c) \leq \alpha(c)$ indicated by Proposition \ref{prop:FPR_le_T1E}, where the latter is typically set at $0.025$ for registration trials). In this case, the ratio $\beta_C(c)/\alpha_B(c) = 40$.
\begin{itemize}
\item Neutral Prior: If the \textit{a priori} probability of effectiveness is $\gamma_1 = 0.5$, the condition $\mathrm{PID}(c) < \alpha(c)$ holds only if $c < 0.9756$.
\item High-Success Area: In a therapeutic area where the probability of effectiveness is high (e.g., $\gamma_1 = 0.7$), the condition is satisfied if $c < 0.9894$ (as illustrated in the right panel of Figure~\ref{fig:oc_normal_design_prior_sd_0.15}, where the PID remains strictly below the Frequentist Type I error rate curve).
\item High-Failure Area: Conversely, in a high-risk therapeutic area where the probability of effectiveness is low (e.g., $\gamma_1 = 0.1$), the Frequentist Type I error only ``protects" the PID if the cutoff is set significantly lower, specifically $c < 0.8163$.
\end{itemize}

The continuous dynamics of this condition are visually captured in Figure~\ref{fig:difference_PID_T1E}, which plots the difference $\mathrm{PID}(c) - \alpha(c)$ across the entire spectrum of prior effectiveness probabilities ($\gamma_1$). The horizontal dashed line at zero represents the boundary condition where the two metrics are identical. The vertical drop lines mark the exact threshold where the equality in Theorem \ref{thm:PID_T1E} is met. To the left of these lines—representing high-risk therapeutic areas with low pre-trial probabilities of success—the curves sit strictly above zero, indicating that the $\mathrm{PID}$ exceeds the nominal Frequentist Type I error rate limit. This illustrates the regulatory safeguard mechanism: to protect the false discovery risk in increasingly challenging disease areas (lower $\gamma_1$), the trial design requires a progressively more stringent posterior cutoff $c$ (e.g., shifting from the $c=0.90$ threshold to $c=0.975$).

This highlights a critical regulatory insight: in high-risk therapeutic areas where the prior probability of effectiveness is low, a conventional $0.975$ cutoff may yield a PID that exceeds the nominal Frequentist Type I error rate. In such scenarios, the Frequentist Type I error rate fails to serve as a conservative upper bound for the actual risk of a false discovery.

The following result further characterizes the bounds of the discrepancy between $\mathrm{PID}(c)$ and $\alpha(c)$, demonstrating that the difference in Bayesian decision reliability relative to the Frequentist Type I error rate is inherently constrained by the posterior cutoff $c$. This result provides a theoretical basis for the ``lower envelope" often observed in simulation studies and underscores how increasingly stringent posterior thresholds limit the achievable divergence between these two metrics.

\begin{proposition}\label{prop:PID_T1E_bounds}
Under large-sample settings, for any posterior cutoff $c \in (0,1)$, the difference between $\mathrm{PID}(c)$ and the Frequentist Type I error rate $\alpha(c)$ is approximately bounded as follows:
\[
-(1-c) \le \mathrm{PID}(c) - \alpha(c) \le c.
\]
Furthermore, at the limits of the prior probability of effectiveness ($\gamma_1$),
\[
\lim_{\gamma_1 \to 1} \bigl[\mathrm{PID}(c) - \alpha(c)\bigr] = -(1-c),
\qquad
\lim_{\gamma_1 \to 0} \bigl[\mathrm{PID}(c) - \alpha(c)\bigr] = c.
\]
\end{proposition}

Figure~\ref{fig:difference_PID_T1E} provides a direct visual confirmation of these theoretical bounds; as the prior probability of effectiveness becomes highly optimistic ($\gamma_1 \to 1$), the difference curves asymptotically flatten out at the exact lower limits dictated by $-(1-c)$.

While the $\mathrm{PID}$ protects the regulatory standard by bounding the risk of false discoveries, comprehensive trial optimization also requires controlling the opportunity cost of false rejections. This is captured by the $\mathrm{FOR}$, the probability that a treatment is genuinely effective given that the trial fails to declare success ($\mathcal{S}=0$). $\mathrm{FOR}$ is closely related to Bayesian conditional power  $\beta_C(c)$ as follows:
\begin{equation*} 
\mathrm{FOR}(c) = \frac{\gamma_1 (1 - \beta_C(c))}{\gamma_1 (1 - \beta_C(c)) + \gamma_0 (1 - \alpha_B(c))}.
\end{equation*}
Thus, increasing the conditional power of the trial, e.g.,  by reducing the threshold $c$,  drives down the $\mathrm{FOR}$. 

This  highlights a fundamental tension in trial design between the regulatory mandate and the sponsor's objective. If a trial designer blindly raises the posterior threshold $c$ to satisfy strict regulatory bounds on the $\mathrm{PID}$ (particularly in high-risk therapeutic areas, as demonstrated by Theorem \ref{thm:PID_T1E}), the trial's conditional power inevitably drops. Consequently, the $\mathrm{FOR}$ sharply inflates, exposing the sponsor to a massive risk of abandoning a genuinely effective, pipeline-defining therapy simply because the decision rule was overly stringent. Therefore, establishing the appropriate posterior threshold $c$ should not rely on evaluating the $\mathrm{PID}$ or the Frequentist Type I error in a vacuum. Instead, trial calibration requires a holistic evaluation of these competing operating characteristics.

\section{Case study} \label{sec:casestudy}

Acute myocardial infarction complicated by cardiogenic shock is a highly lethal condition that places patients at extreme risk for early death \citep{hochman1999early, hochman2006early, thiele2015management}. While emergency revascularization is the established cornerstone of treatment, a major clinical dilemma persisted regarding whether physicians should intervene solely on the infarct-related artery or concurrently treat all significant coronary blockages. To resolve this debate, the CULPRIT-SHOCK trial evaluated these revascularization strategies by randomly assigning 706 patients to either culprit-lesion-only percutaneous coronary intervention (PCI) or immediate multivessel PCI. At one year, \citet{thiele2018one} reported that the primary composite endpoint of death or renal-replacement therapy occurred in 50.0\% (172 of 344) of the culprit-only group and 56.9\% (194 of 341) of the multivessel group. Yielding a relative risk of 0.88 with a 95\% confidence interval (CI) of 0.76 to 1.01, the upper bound marginally crossed the null threshold, with the corresponding $p$-value $> 0.05$.  As a result, the authors concluded that “mortality did not differ significantly between the two groups at 1 year of follow-up.”

We retrospectively re-designed the CULPRIT-SHOCK trial using the Bayesian framework and principles established in the preceding sections. The treatment effect is parameterized as the absolute risk reduction, $\theta = \theta_C - \theta_T$, where $\theta_C$ and $\theta_T$ denote the 1-year mortality rates for the multivessel control and culprit-only arms, respectively. The success criterion for concluding that the culprit-only arm is superior to the control is defined as $\Pr(\theta > 0 \mid \mathcal{D}) > c$. Below, we evaluate two approaches for calibrating the threshold $c$.

The first approach calibrates $c$ by controlling the one-sided Frequentist Type I error rate at a nominal level of $\alpha^* = 0.025$. Under non-informative analysis priors for the treatment and control arms ($\pi_{\mathrm{a}, T}(\theta_T) = \text{Beta}(1, 1)$ and $\pi_{\mathrm{a}, C}(\theta_C) = \text{Beta}(1, 1)$) with beta-binomial models, the calibrated threshold is $c = 0.975$ (Table~\ref{tbl:case_study_calibrated_final}). Given the observed trial data, the posterior probability of efficacy is $\Pr(\theta > 0 \mid \mathcal{D}) = 0.965$; consequently, the trial fails to meet the success criterion. This outcome is expected, as a decision rule calibrated to the Frequentist Type I error rate, particularly when employing non-informative analysis priors, naturally yields results consistent with standard frequentist significance testing.

The second approach adopts a Bayesian paradigm, calibrating $c$ by setting the PID to a nominal value of $\tau^* = 0.025$. As previously noted, the PID is intrinsically linked to the choice of the design prior. We constructed the design prior based on historical registry data \citep{van2010effect}, which reported 1-year mortality rates of 53.0\% (66/124) for culprit-only and 60.0\% (22/37) for multivessel PCI. By updating a non-informative prior $\text{Beta}(1, 1)$ with these data, we obtained design priors of $\pi_{\mathrm{d}, T}(\theta_T) = \text{Beta}(67, 59)$ and $\pi_{\mathrm{d}, C}(\theta_C) = \text{Beta}(23, 16)$ for the treatment and control arms, respectively. While we acknowledge that alternative specifications for the design prior exist, this approach provides a principled and evidence-based foundation for evaluating operating characteristics and performing robust sensitivity analyses.

First, we consider a scenario where the sponsor employs analysis priors that perfectly match the design priors for each arm, specifically $\pi_{\mathrm{a}, T}(\theta_T) = \text{Beta}(67, 59)$ and $\pi_{\mathrm{a}, C}(\theta_C) = \text{Beta}(23, 16)$. A numerical calibration targeting a $\mathrm{PID}$ of 0.025 lowers the required success threshold to $c = 0.8145$. 
%This result empirically validates Corollary \ref{cor:PID_bound}, demonstrating that under perfectly matched priors, the PID is strictly bounded above by $1 - c = 0.1855$. 
Driven by these informative analysis priors, the trial's posterior probability increases to $\Pr(\theta > 0 \mid \mathcal{D}) = 0.966$, successfully meeting the efficacy criterion. Under this setting, the BCP is 83.2\% and the FOR is 34.1\%. For reference, the Frequentist Type I error rate is 21.8\%, though we note that this calibration procedure does not explicitly aim to control frequentist error. 

We then consider a scenario utilizing non-informative analysis priors, $\pi_{\mathrm{a}, T}(\theta_T) = \text{Beta}(1, 1)$ and $\pi_{\mathrm{a}, C}(\theta_C) = \text{Beta}(1, 1)$, to minimize the influence of the prior on the final inference. In this case, the calibrated threshold is $c = 0.772$ to maintain the PID at $\tau^* = 0.025$. Given that the observed trial data yielded a posterior probability of 0.965, the threshold is exceeded, and success is again declared. We further evaluated a more conservative target risk of $\tau^* = 0.01$, which increased the threshold to $c = 0.909$; notably, the final decision remains unchanged. Under this stricter target, the BCP and FOR are 71.6\% and 45.3\%, respectively (Table~\ref{tbl:case_study_calibrated_final}).

Finally, we conducted a sensitivity analysis by systematically varying the treatment arm's design prior $\pi_{\mathrm{d}, T}(\theta_T)$ while holding its effective sample size constant. Specifically, we examined a neutral scenario ($\pi_{\mathrm{d}, T}(\theta_T) = \text{Beta}(74, 52)$), where mean mortality matches the control at roughly 59.0\%, and a pessimistic scenario ($\pi_{\mathrm{d}, T}(\theta_T) = \text{Beta}(81, 45)$), where the treatment underperforms the control by the same 5.8\% margin observed in the optimistic setting (mean mortality $\approx$ 64.8\%). The design prior for the control arm remains fixed at $\pi_{\mathrm{d}, C}(\theta_C) =\text{Beta}(23, 16)$. As shown in Table \ref{tbl:case_study_calibrated_final}, the CULPRIT-SHOCK trial consistently meets the efficacy criterion across the majority of PID targets and design prior scenarios. The sole exception occurs under the most extreme combination: the pessimistic design prior coupled with a strict $\tau^* = 0.01$. 

Overall, the Bayesian analysis provides robust evidence for the clinical advantage of culprit-only PCI. Our conclusion aligns with the interpretation by \citet{Yan2019}, who noted that while classic Type I error approaches may conclude that mortality did not differ significantly, a Bayesian analysis demonstrates that the odds are approximately 28 to 1 in favor of culprit-only PCI—a margin that most clinicians and patients would find compelling for treatment selection.

\section{Discussion} \label{sec:Discussion}
In this work, we formalize various Bayesian decision-error metrics to provide a rigorous demonstration of error control in clinical trials based on Bayesian methodology \citep{FDA2026, EMA2026}. We systematically investigate the theoretical interrelationships of these metrics and their alignment with the ``standard" approach of calibrating Bayesian success criteria to the Frequentist Type I error rate. We demonstrate that while these Bayesian decision-error metrics are related to the latter, they are fundamentally distinct. In particular, calibrating to the PID may be more conservative, less conservative, or equivalent to calibrating to the Frequentist Type I error rate, depending critically on the congruence between the design prior and the analysis prior. Our theoretical results provide a robust roadmap for calibrating Bayesian decision rules in a transparent and reproducible manner.

While our analysis focused on Bayesian success criteria based on posterior probability, the underlying principles should hold for other related criteria, such as those based on predictive probability, as the posterior distribution encapsulates all relevant information regarding $\theta$ contained in the data. Furthermore, while we have prioritized operating characteristics related to decision-making, recent regulatory guidance emphasizes that other inferential metrics, including bias, mean square error, and the coverage of credible intervals, must also be rigorously evaluated during the trial design phase. Under the framework presented here, a sensible approach is to first calibrate the posterior probability threshold $c$ based on a primary decision-error metric (e.g., Frequentist Type I error or PID) and subsequently evaluate inferential metrics given that threshold. If the resulting metrics reveal suboptimal operating characteristics, the value of $c$ can be iteratively refined to achieve a well-balanced profile across all statistical dimensions.

Our analysis has primarily focused on the relationship between Bayesian calibration methods based on Frequentist Type I error control and PID control. A natural question arises regarding the relationship between the resulting Bayesian success criteria and the classic Frequentist hypothesis testing framework. We briefly explore this connection within the Supplementary Materials (Section \ref{proofprop3}). Specifically, we demonstrate that a Bayesian posterior success criterion  (\ref{eqn:decision_rule}) calibrated for Frequentist Type I error is asymptotically equivalent to a Frequentist hypothesis test. Furthermore, we show that in certain cases, such as under a normal model, these two frameworks are exactly equivalent.

\newpage

\begin{table}[ht]
\centering
\caption{Decision outcomes under the Bayesian success rule.}
\label{tbl:2by2}
\begin{tabular}{lcc} 
\toprule
 & $\theta>\delta$ (effective) & $\theta\le\delta$ (ineffective) \\
\midrule
Success & True positive(TP) & False positive (FP) \\
Failure & False negative (FN) & True negative (TN) \\
\bottomrule
\end{tabular}
\end{table}

\begin{landscape}
\begin{table}[ht] 
\centering
\caption{Summary of operating characteristic metrics, their definitions, analytical solutions, and clinical interpretations.}
\label{tbl:OC_summary}
\renewcommand{\arraystretch}{1.8}

\begin{tabular}{
>{\centering\arraybackslash}m{3.8cm}
>{\centering\arraybackslash}m{3.2cm}
>{\centering\arraybackslash}m{3cm}
>{\centering\arraybackslash}m{4cm}
>{\arraybackslash}m{6.5cm}
}

\toprule
Metric & Probability Definition & Decision-table Interpretation & Analytical Solution & \multicolumn{1}{>{\centering\arraybackslash}m{6cm}}{Clinical Interpretation} \\
\midrule

Bayesian conditional power $\beta_C(c)$
& $\Pr(\mathcal{S}=1 \mid \mathcal{E})$
& $\displaystyle \frac{\mathrm{TP}}{\mathrm{TP}+\mathrm{FN}}$
& $\displaystyle 1 - \frac{ \Phi(b) - \Phi_2(a, b; \rho)}{1 - \Phi(a)}$
& Among truly effective treatments, the chance the trial declares success. \\

Bayesian power $\beta_B(c)$
& $\Pr(\mathcal{S}=1)$
& $\displaystyle \frac{\mathrm{TP}+\mathrm{FP}}{N}$
& $\displaystyle 1 - \Phi(b)$
& Predictive probability of declaring success before seeing data, averaged over all plausible effects. \\

(Frequentist) Type I error rate $\alpha(c)$
& $\Pr(\mathcal{S}=1 \mid \theta=\delta)$
&
& $\displaystyle 1 - \Phi\left(\frac{y_c - \delta}{v} \right)$
& Worst-case false-positive guarantee at the null boundary, independent of the design prior. \\

Bayesian false positive rate $\alpha_B(c)$
& $\Pr(\mathcal{S}=1 \mid \bar{\mathcal{E}})$
& $\displaystyle \frac{\mathrm{FP}}{\mathrm{FP}+\mathrm{TN}}$
& $\displaystyle 1 - \frac{ \Phi_2(a, b; \rho)}{\Phi(a)}$
& Among truly ineffective treatments, the chance the trial incorrectly declares success. \\

Probability of incorrect decision $\mathrm{PID}(c)$
& $\Pr(\bar{\mathcal{E}} \mid \mathcal{S}=1)$
& $\displaystyle \frac{\mathrm{FP}}{\mathrm{TP}+\mathrm{FP}}$
& $\displaystyle \frac{\Phi(a) - \Phi_2(a, b; \rho)}{1 - \Phi(b)}$
& Among declared successes, the probability the treatment is actually ineffective. \\

False omission rate $\mathrm{FOR}(c)$
& $\Pr(\mathcal{E} \mid \mathcal{S}=0)$
& $\displaystyle \frac{\mathrm{FN}}{\mathrm{FN}+\mathrm{TN}}$
& $\displaystyle 1 - \frac{\Phi_2(a, b; \rho)}{\Phi(b)}$
& Among declared failures, the probability the treatment is actually effective. \\

\bottomrule
\end{tabular}

\vspace{-2mm}
\begin{flushleft}\footnotesize
\textbf{Notes.} $\mathcal{E}$ denotes the effectiveness event ($\theta > \delta$) and $\mathcal{S} = 1$ indicates trial success. 
$N = \mathrm{TP}+\mathrm{FP}+\mathrm{FN}+\mathrm{TN}$ denotes the total probability space. 
Closed-form expressions are written in terms of the standard normal cumulative distribution function $\Phi(\cdot)$ and the bivariate normal cumulative distribution function $\Phi_2(\cdot,\cdot;\rho)$. 
Here, $y_c$ denotes the decision boundary for the sample mean corresponding to the posterior probability cutoff $c$, $a$ and $b$ denote the standardized boundaries corresponding to the clinical margin and decision threshold, respectively, and $\rho$ denotes the correlation between the true treatment effect and the sample mean under the design prior. $v^2$ denotes the sampling variance of the sample mean estimator. Derivations of the analytical expressions are provided in the Supplementary Materials.
\end{flushleft}

\end{table}
\end{landscape}

\begin{figure}[H]
\centering
\includegraphics[width=1\textwidth]{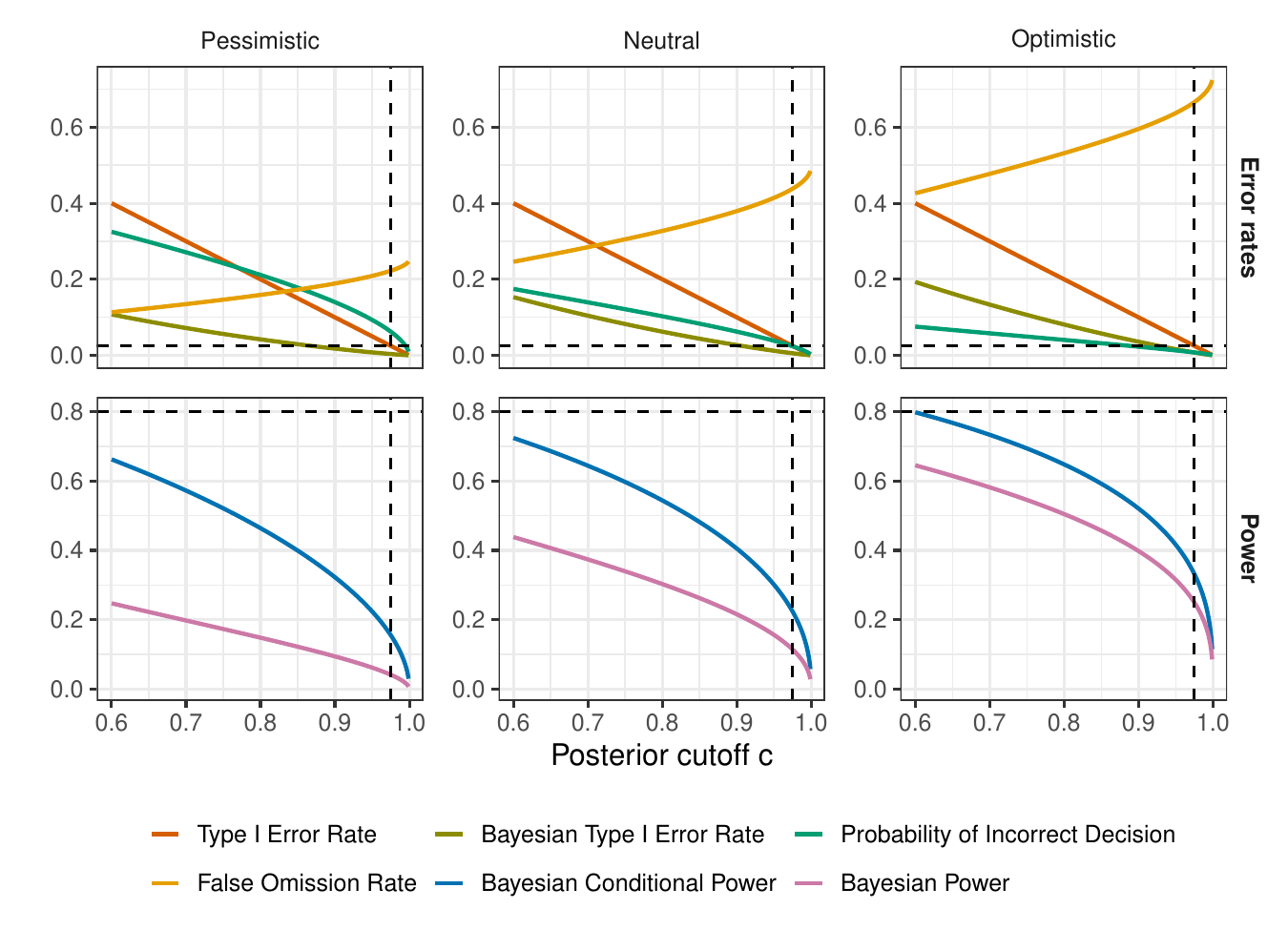}
\caption{Operating characteristics for single-arm continuous endpoints across posterior probability cutoffs $c$, evaluated under three design priors: pessimistic ($\pi_{\mathrm{d},T}(\theta_T) = N(-0.1, 0.15^2)$), neutral ($\pi_{\mathrm{d},T}(\theta_T) = N(0, 0.15^2)$), and optimistic ($\pi_{\mathrm{d},T}(\theta_T) = N(0.1, 0.15^2)$), corresponding to $\gamma_1 = 0.252$, $0.5$, and $0.748$, respectively. The vertical dashed line indicates $c = 0.975$. The horizontal dashed lines in the error-rate and power panels indicate reference levels of $0.025$ and $0.80$, respectively. The trial design assumes a clinical margin of $\delta = 0$, a sample size of $n_T = 74$, variance $\sigma = 1$, and a non-informative analysis prior ($\pi_{\mathrm{a}, T}(\theta_T) = N(0, 10^6)$).}
\label{fig:oc_normal_design_prior_sd_0.15}
\end{figure}

\begin{figure}[H]
    \centering
    \includegraphics[width=0.5\textwidth]{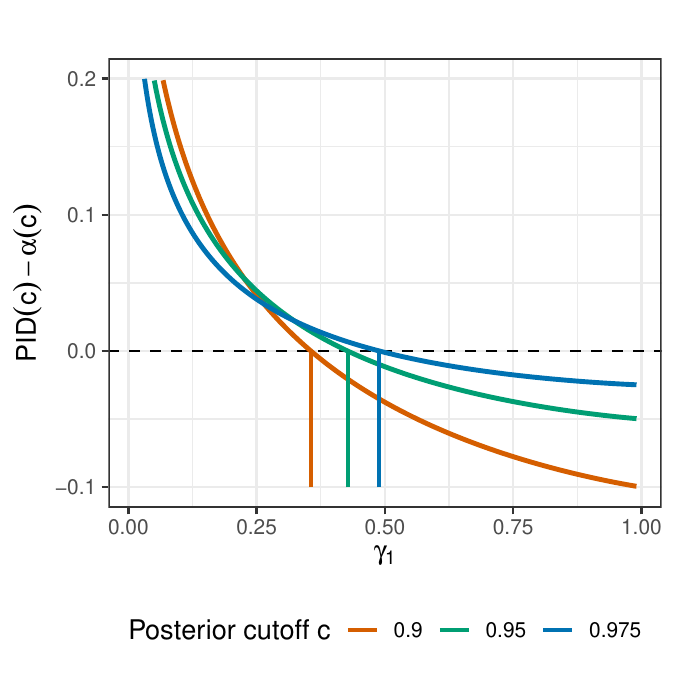}
    \caption{Difference between the probability of an incorrect decision (PID) and the Frequentist Type~I error across posterior probability cutoffs $c = 0.90$, $0.95$, and $0.975$, under design priors with standard deviation equals $0.15$. The x-axis, $\gamma_1$, represents the prevalence of effective trials under the design prior. Vertical dashed lines indicate the threshold values of $\gamma_1$ at which the PID transitions from exceeding to falling below the Type~I error rate.}
    \label{fig:difference_PID_T1E}
\end{figure}

% \end{landscape}

\begin{landscape}
\begin{table}[ht]
\centering
\caption{Operating characteristics for the CULPRIT-SHOCK case study under different PID calibration targets.}
\label{tbl:case_study_calibrated_final}
\renewcommand{\arraystretch}{1.3}
\begin{tabular}{@{} c c c c ccccc @{}}
\toprule
\multirow{2}{*}{\makecell[c]{Calibration \\ Target}} & 
\multirow{2}{*}{\makecell[c]{Analysis Prior \\ $(\pi_{\mathrm{a}}(\theta_T, \theta_C))$}} & 
\multirow{2}{*}{\makecell[c]{Design Prior$^{\dagger}$ \\ $(\pi_{\mathrm{d},T}(\theta_T))$}} & 
\multirow{2}{*}{\makecell[c]{Calibrated \\ $c$}} & 
\multicolumn{5}{c}{Metrics} \\
\cmidrule(lr){5-9}
& & & & FT1E & BT1E & FOR & BCP & BP \\
\midrule

% Benchmark Block (Single Row)
\makecell[c]{FT1E $=0.025$} 
& \makecell[c]{$\pi_{\mathrm{a},T}(\theta_T)=\mathrm{Beta}(1,1)$ \\ $\pi_{\mathrm{a},C}(\theta_C)=\mathrm{Beta}(1,1)$}
& -- 
& 0.975 & 0.025 & -- & -- & -- & -- \\
\midrule

% Block 1 (Single Row)
\makecell[c]{PID $=0.025$} 
& \makecell[c]{$\pi_{\mathrm{a},T}(\theta_T)=\mathrm{Beta}(67,59)$ \\ $\pi_{\mathrm{a},C}(\theta_C)=\mathrm{Beta}(23,16)$}
& \makecell[c]{Historical \\ $\pi_{\mathrm{d},T}(\theta_T)=\mathrm{Beta}(67, 59)$} 
& 0.8145 & 0.218 & 0.061 & 0.341 & 0.832 & 0.633 \\
\midrule

% Block 2 (Three Rows)
& 
& \makecell[c]{Historical \\ $\pi_{\mathrm{d},T}(\theta_T)=\mathrm{Beta}(67, 59)$} & 0.772 & 0.221 & 0.060 & 0.357 & 0.817 & 0.621 \\ \cline{3-9}
\makecell[c]{PID $=0.025$} & 
\makecell[c]{$\pi_{\mathrm{a},T}(\theta_T)=\mathrm{Beta}(1,1)$ \\ $\pi_{\mathrm{a},C}(\theta_C)=\mathrm{Beta}(1,1)$} & 
\makecell[c]{Neutral \\ $\pi_{\mathrm{d},T}(\theta_T)=\mathrm{Beta}(74, 52)$} & 0.898 & 0.102 & 0.017 & 0.290 & 0.620 & 0.327 \\  \cline{3-9}
& 
& \makecell[c]{Pessimistic \\ $\pi_{\mathrm{d},T}(\theta_T)=\mathrm{Beta}(81, 45)$} & 0.954 & 0.046 & 0.004 & 0.187 & 0.410 & 0.117 \\
\midrule

% Block 3 (Three Rows)
& 
& \makecell[c]{Historical \\ $\pi_{\mathrm{d},T}(\theta_T)=\mathrm{Beta}(67, 59)$} & 0.909 & 0.095 & 0.021 & 0.453 & 0.716 & 0.536 \\  \cline{3-9}
\makecell[c]{PID $=0.010$} & 
\makecell[c]{$\pi_{\mathrm{a},T}(\theta_T)=\mathrm{Beta}(1,1)$ \\ $\pi_{\mathrm{a},C}(\theta_C)=\mathrm{Beta}(1,1)$} & 
\makecell[c]{Neutral \\ $\pi_{\mathrm{d},T}(\theta_T)=\mathrm{Beta}(74, 52)$} & 0.960 & 0.040 & 0.005 & 0.345 & 0.505 & 0.262 \\  \cline{3-9}
& 
& \makecell[c]{Pessimistic \\ $\pi_{\mathrm{d},T}(\theta_T)=\mathrm{Beta}(81, 45)$} & 0.983 & 0.017 & 0.001 & 0.211 & 0.309 & 0.087 \\
\bottomrule
\end{tabular}

\vspace{1mm}
\begin{flushleft}
\small \textbf{Note}: FT1E: Frequentist Type I error rate; BT1E: Bayesian Type I error rate; FOR: False omission rate; BCP: Bayesian conditional power; BP: Bayesian power. \\
$^{\dagger}$The design priors displayed in the table correspond to the treatment arm, while the design prior for the control arm is fixed at $\pi_{\mathrm{d},C} = \mathrm{Beta}(23,16)$ across all scenarios.
\end{flushleft}
\end{table}
\end{landscape}

\clearpage
\bibliographystyle{plainnat}

% number 99 determines how many citations can be included in the file (maximum 99)

\newpage

\appendix
% \appendix
\section*{Supplementary Materials}
% Reset counters for the supplement
\setcounter{table}{0}
\setcounter{figure}{0}
\setcounter{equation}{0}
\setcounter{section}{0}

% Add 'S' prefix to all numbers
\renewcommand{\thetable}{S\arabic{table}}
\renewcommand{\thefigure}{S\arabic{figure}}
\renewcommand{\theequation}{S\arabic{equation}}
\renewcommand{\thesection}{S\arabic{section}}

\section{Proofs}

\subsection{Proof of Proposition~\ref{prop:BP_decomp_CP}}

\begin{proof}
By the law of total probability, the marginal probability of meeting the success criterion, defined as Bayesian power $\beta_B(c) = \Pr(\mathcal{S}=1)$, can be partitioned over the two complementary and disjoint states of nature: effectiveness $\mathcal{E} = \{\theta > \delta\}$ and ineffectiveness $\overline{\mathcal{E}} = \{\theta \le \delta\}$. Specifically:
\begin{equation*}
\Pr(\mathcal{S}=1) = \Pr(\mathcal{S}=1 \mid \mathcal{E})\Pr(\mathcal{E}) + \Pr(\mathcal{S}=1 \mid \overline{\mathcal{E}})\Pr(\overline{\mathcal{E}}).
\end{equation*}
Based on the definitions, we have $\beta_C(c) = \Pr(\mathcal{S}=1 \mid \mathcal{E})$ as the Bayesian conditional power and $\alpha_B(c) = \Pr(\mathcal{S}=1 \mid \overline{\mathcal{E}})$ as the Bayesian Type I error rate. Furthermore, the design prior $\pi_{\mathrm{d}}(\theta)$ induces the prior prevalence of effective and ineffective treatments, denoted by $\gamma_1 = \Pr(\mathcal{E})$ and $\gamma_0 = \Pr(\overline{\mathcal{E}})$, respectively. Substituting these terms into the expression above yields the decomposition:
\begin{equation*}
\beta_B(c) = \gamma_1 \beta_C(c) + \gamma_0 \alpha_B(c).
\end{equation*}
In practical clinical trial settings where the design maintains a low false positive risk and high success probability, we typically observe $\alpha_B(c) < \beta_C(c)$. Under this condition, it follows mathematically that $\beta_B(c) \le \beta_C(c)$, with equality $\beta_B(c) \to \beta_C(c)$ occurring only as the prior prevalence of effectiveness $\gamma_1 \to 1$.
\end{proof}

\subsection{Proof of Proposition~\ref{prop:FPR_le_T1E}}

\begin{proof}
The Frequentist Type I error rate $\alpha(c)$ is defined strictly at the point null hypothesis boundary $\theta = \delta$:
\begin{equation*}
\alpha(c) = \Pr\{\mathcal{S}(\mathcal{D};c)=1 \mid \theta = \delta\}.
\end{equation*}
In contrast, the Bayesian Type I error rate $\alpha_B(c)$ is defined as the probability of success conditioned on the treatment being ineffective ($\theta \le \delta$), which averages the success probability over the entire ineffectiveness region according to the design prior $\pi_{\mathrm{d}}(\theta)$ restricted to $\{\theta \le \delta\}$:
\begin{equation*}
\alpha_B(c) = \int_{-\infty}^{\delta} \Pr\{\mathcal{S}(\mathcal{D};c)=1 \mid \theta\} \pi_{\mathrm{d}}(\theta \mid \theta \le \delta) d\theta.
\end{equation*}
Under the standard assumption that the power function $\beta(\theta) = \Pr\{\mathcal{S}(\mathcal{D};c)=1 \mid \theta\}$ is a monotonically increasing function of $\theta$ (i.e., the probability of declaring success increases as the treatment effect improves), it follows that for all $\theta < \delta$, $\beta(\theta) < \beta(\delta) = \alpha(c)$. Consequently:
\begin{equation*}
\alpha_B(c) = \int_{-\infty}^{\delta} \beta(\theta) \pi_{\mathrm{d}}(\theta \mid \theta \le \delta) d\theta < \int_{-\infty}^{\delta} \alpha(c) \pi_{\mathrm{d}}(\theta \mid \theta \le  \delta) d\theta = \alpha(c).
\end{equation*}
The equality $\alpha_B(c) = \alpha(c)$ holds if and only if the design prior $\pi_{\mathrm{d}}(\theta)$ is a point mass concentrated entirely at the boundary point $\theta = \delta$. In all other practical settings where the design prior explores the ineffectiveness region, the Bayesian Type I error rate remains strictly bounded by the Frequentist Type I error rate: $\alpha_B(c) \le \alpha(c)$.
\end{proof}

\subsection{Proof of Proposition~\ref{prop:alpha_asymp}} \label{proofprop3}

\begin{proof}
Under standard regularity conditions such that the Bernstein--von Mises theorem \citep{van2000asymptotic} holds, the posterior distribution of $\theta$ is asymptotically normal. In particular, when the true value is at the boundary $\theta=\delta$,
\[
\theta \mid \mathcal{D}
\approx
N\!\left(\hat{\theta},\, \frac{I(\delta)^{-1}}{n}\right),
\]
where $\hat{\theta}$ is a consistent asymptotically normal estimator of $\theta$, and $I(\delta)$ is the Fisher information at $\theta=\delta$. Therefore,
\[
\Pr(\theta>\delta\mid\mathcal{D})
\approx
1-\Phi\!\left(
\frac{\delta-\hat{\theta}}{\sqrt{I(\delta)^{-1}/n}}
\right)
=
\Phi\!\left(
\frac{\hat{\theta}-\delta}{\sqrt{I(\delta)^{-1}/n}}
\right).
\]
Hence, the Bayesian success rule
\[
\mathcal{S}(\mathcal{D};c)=
\mathbb{I}\!\left\{\Pr(\theta>\delta\mid\mathcal{D})>c\right\}
\]
is asymptotically equivalent to
\[
\mathbb{I}\!\left\{
\frac{\hat{\theta}-\delta}{\sqrt{I(\delta)^{-1}/n}}
> z_c
\right\},
\]
where $z_c=\Phi^{-1}(c)$. Since $\alpha(c)$ is evaluated at $\theta=\delta$, the asymptotic normality of $\hat{\theta}$ implies
\[
\frac{\hat{\theta}-\delta}{\sqrt{I(\delta)^{-1}/n}}
\overset{d}{\to}
N(0,1).
\]
It follows that
\[
\alpha(c)
=
\Pr\{\mathcal{S}(\mathcal{D};c)=1\mid \theta=\delta\}
\to
\Pr(Z>z_c)
=
1-c,
\]
where $Z\sim N(0,1)$.

This result also demonstrates that the Bayesian posterior success criterion (1) calibrated for
Frequentist Type I error is asymptotically equivalent to a Frequentist hypothesis test. Furthermore,
as shown next, under a normal model, these two frameworks are exactly equivalent.

\end{proof}

\subsubsection{Example: Connection to the one-sided $t$-test with a continuous endpoint} \label{sec:normal_example}

Consider a single-arm trial with a continuous endpoint, where outcomes $Y_1, \dots, Y_n \mid \theta, \sigma^2 \overset{\text{iid}}{\sim} N(\theta, \sigma^2)$. We evaluate the success rule $\Pr(\theta > \delta \mid \mathcal{D}) > c$ under the non-informative Jeffreys prior $\pi(\theta, \sigma^2) \propto 1/\sigma^2$. The joint posterior distribution is derived via Bayes' rule as:
\begin{align*}
    \pi(\theta, \sigma^2 \mid \mathcal{D}) 
    &\propto (\sigma^2)^{-(n/2+1)} \exp\left\{ -\frac{1}{2\sigma^2} \sum_{i=1}^n (Y_i - \theta)^2 \right\} \\
    &\propto (\sigma^2)^{-(n/2+1)} \exp\left\{ -\frac{(n-1)S^2 + n(\bar{Y} - \theta)^2}{2\sigma^2} \right\},
\end{align*}
where $\bar{Y} = n^{-1} \sum_{i=1}^n Y_i$ and $S^2 = (n-1)^{-1} \sum_{i=1}^n (Y_i - \bar{Y})^2$ are the sample mean and variance, respectively. 

To draw inference on $\theta$, we marginalize over $\sigma^2$, yielding the Student-$t$ posterior:
\[
\theta \mid \mathcal{D} \sim t_{n-1}\left(\bar{Y}, \frac{S}{\sqrt{n}}\right), \quad \text{or equivalently,} \quad \frac{\theta - \bar{Y}}{S/\sqrt{n}} \mid \mathcal{D} \sim t_{n-1}.
\]
The posterior probability of effectiveness is thus:
\[
\Pr(\theta > \delta \mid \mathcal{D}) = 1 - F_{t_{n-1}}\left( \frac{\delta - \bar{Y}}{S/\sqrt{n}} \right) = F_{t_{n-1}}\left( \frac{\bar{Y} - \delta}{S/\sqrt{n}} \right),
\]
where $F_{t_{n-1}}$ denotes the cumulative distribution function of the Student-$t$ distribution with $n-1$ degrees of freedom. It follows that the Bayesian success rule $\Pr(\theta > \delta \mid \mathcal{D}) > c$ is equivalent to:
\[
\frac{\bar{Y} - \delta}{S/\sqrt{n}} > t_{n-1, c},
\]
where $t_{n-1, c}$ is the $c$-th quantile of the $t_{n-1}$ distribution. This is mathematically identical to the rejection region of the classical one-sided one-sample $t$-test at significance level $\alpha^\ast = 1 - c$.

\subsection{Proof of Theorem~\ref{thm:PID_bound}}

\begin{proof}

Let $\boldsymbol{\theta}$ denote the parameter space for the trial, where $\boldsymbol{\theta} = \theta_T$ for a single-arm design and $\boldsymbol{\theta} = (\theta_T, \theta_C)$ for a two-arm randomized controlled trial. We assume $\boldsymbol{\theta}$ is governed by a design prior $\pi_{\mathrm{d}}(\boldsymbol{\theta})$ and an analysis prior $\pi_{\mathrm{a}}(\boldsymbol{\theta})$, which factorize as $\pi_T(\theta_T)\pi_C(\theta_C)$ in the independent two-arm setting. To derive the bound for $\mathrm{PID}(c)$, we first express the posterior probability of ineffectiveness under the design prior $\pi_{\mathrm{d}}$ as

\begin{align}
\Pr_{\mathrm{d}}(\bar{\mathcal{E}} \mid \mathcal{D})
&=
\frac{f(\mathcal{D} \mid \bar{\mathcal{E}})\Pr_{\mathrm{d}}(\bar{\mathcal{E}})}
     {f(\mathcal{D} \mid \bar{\mathcal{E}})\Pr_{\mathrm{d}}(\bar{\mathcal{E}})
      + f(\mathcal{D} \mid \mathcal{E})\Pr_{\mathrm{d}}(\mathcal{E})} \nonumber \\
&=
\frac{1}{\frac{\Pr_{\mathrm{d}}(\mathcal{E})}{\Pr_{\mathrm{d}}(\bar{\mathcal{E}})}
        \frac{f(\mathcal{D} \mid \mathcal{E})}{f(\mathcal{D} \mid \bar{\mathcal{E}})} + 1}
=
\frac{1}{OE_{\mathrm{d}} \cdot LR(\mathcal{D}) + 1},
\label{eqn:PID_posterior}
\end{align}
where $OE_{\mathrm{d}} = \Pr_{\mathrm{d}}(\mathcal{E}) / \Pr_{\mathrm{d}}(\bar{\mathcal{E}}) = \gamma_1/\gamma_0$ denotes the prior odds of effectiveness under the design prior and
\[
LR(\mathcal{D}) =
\frac{f(\mathcal{D}\mid\mathcal{E})}{f(\mathcal{D}\mid\bar{\mathcal{E}})}
\]
is the likelihood ratio comparing the effective and ineffective states. Under the conditional equivalence assumption stated in Theorem~\ref{thm:PID_bound}, which applies analogously to the joint parameter space such that
\[
\pi_{\mathrm{a}}(\boldsymbol{\theta} \mid \mathcal{E}) = \pi_{\mathrm{d}}(\boldsymbol{\theta} \mid \mathcal{E}),
\qquad
\pi_{\mathrm{a}}(\boldsymbol{\theta} \mid \bar{\mathcal{E}}) = \pi_{\mathrm{d}}(\boldsymbol{\theta} \mid \bar{\mathcal{E}}),
\]
the likelihood ratio can be written equivalently under either prior as
\[
LR(\mathcal{D})
=
\frac{\int_{\mathcal{E}} f(\mathcal{D}\mid\boldsymbol{\theta})\pi_{\mathrm{a}}(\boldsymbol{\theta}\mid\mathcal{E})d\boldsymbol{\theta}}
     {\int_{\bar{\mathcal{E}}} f(\mathcal{D}\mid\boldsymbol{\theta})\pi_{\mathrm{a}}(\boldsymbol{\theta}\mid\bar{\mathcal{E}})d\boldsymbol{\theta}}
=
\frac{\int_{\mathcal{E}} f(\mathcal{D}\mid\boldsymbol{\theta})\pi_{\mathrm{d}}(\boldsymbol{\theta}\mid\mathcal{E})d\boldsymbol{\theta}}
     {\int_{\bar{\mathcal{E}}} f(\mathcal{D}\mid\boldsymbol{\theta})\pi_{\mathrm{d}}(\boldsymbol{\theta}\mid\bar{\mathcal{E}})d\boldsymbol{\theta}}.
\]
Let the success region be defined as $\mathcal{R}_c = \{\mathcal{D} : \mathcal{S}(\mathcal{D};c) = 1\}$, representing the set of all possible datasets for which the trial concludes in success. For any dataset $\mathcal{D} \in \mathcal{R}_c$, the posterior probability of effectiveness exceeds the threshold $c$, and the posterior odds under the analysis prior satisfy:
\[
\frac{\Pr_{\mathrm{a}}(\mathcal{E}\mid\mathcal{D})}
     {\Pr_{\mathrm{a}}(\bar{\mathcal{E}}\mid\mathcal{D})}
=
OE_{\mathrm{a}} \cdot LR(\mathcal{D})
\ge
\frac{c}{1-c},
\]
where $OE_{\mathrm{a}} = \Pr_{\mathrm{a}}(\mathcal{E})/\Pr_{\mathrm{a}}(\bar{\mathcal{E}})$ denotes the prior odds under the analysis prior. This yields the lower bound
\begin{equation}
LR(\mathcal{D}) \ge \frac{1}{OE_{\mathrm{a}}}\frac{c}{1-c}.
\label{eqn:LR_lower_bound}
\end{equation}
Substituting \eqref{eqn:LR_lower_bound} into \eqref{eqn:PID_posterior} and defining the prior odds ratio $R = \frac{OE_{\mathrm{d}}}{OE_{\mathrm{a}}}$, we obtain
\[
\Pr_{\mathrm{d}}(\bar{\mathcal{E}}\mid\mathcal{D})
\le
\frac{1}{R\left(\frac{c}{1-c}\right)+1}.
\]
Finally, by the law of total expectation,
\[
\mathrm{PID}(c) = \Pr_{\mathrm{d}}(\bar{\mathcal{E}} \mid \mathcal{S}=1)
=
\mathbb{E}_{\mathrm{d}}
\left[
\Pr_{\mathrm{d}}(\bar{\mathcal{E}}\mid\mathcal{D})
\mid
\mathcal{D}\in\mathcal{R}_c
\right].
\]
Since $\Pr_{\mathrm{d}}(\bar{\mathcal{E}}\mid\mathcal{D})$ is bounded above by the constant derived above for every dataset in $\mathcal{R}_c$, the same bound holds for the expectation:
\[
\mathrm{PID}(c)
\le
\frac{1}{R\left(\frac{c}{1-c}\right)+1}.
\]
\end{proof}

\subsection{Proof of Corollary~\ref{cor:PID_bound}}
\begin{proof}
Corollary~\ref{cor:PID_bound} considers the specific case where the analysis prior and the design prior are perfectly aligned, such that $\pi_{\mathrm{a}}(\theta) = \pi_{\mathrm{d}}(\theta)$. Under this condition, the prior odds of effectiveness are identical, $OE_{\mathrm{a}} = OE_{\mathrm{d}}$, which implies that the prior odds ratio is $R = 1$. Substituting $R=1$ into the general bound established in Theorem 1 yields:
\begin{equation*}
    \mathrm{PID}(c) \le \frac{1}{1 \cdot \left( \frac{c}{1 - c} \right) + 1} = \frac{1}{\frac{c + 1 - c}{1 - c}} = 1 - c \text{.}
\end{equation*}
This demonstrates that when the analysis and design priors are congruent, the PID is strictly bounded by the complement of the posterior threshold, $1-c$. 
\end{proof}

\subsection{Proof of Theorem~\ref{thm:PID_T1E}}

\begin{proof}
Recall that
\[
\mathrm{PID}(c)
=
\frac{\gamma_0 \alpha_B(c)}
     {\gamma_0 \alpha_B(c) + \gamma_1 \beta_C(c)},
\]
where $\gamma_1 = \Pr(\theta > \delta)$ and $\gamma_0 = 1-\gamma_1$ denote the prior prevalence of effective and ineffective treatments under the design prior, respectively. Under large-sample settings, the posterior distribution is asymptotically normal, and the posterior success rule $\Pr(\theta > \delta \mid D) > c$ is approximately equivalent to a one-sided large-sample test \citep{van2000asymptotic}. Consequently, the corresponding frequentist Type~I error rate satisfies $\alpha(c) \approx 1-c$. Therefore, to characterize when $\mathrm{PID}(c) < \alpha(c)$, it suffices asymptotically to study the condition $\mathrm{PID}(c) < 1-c$. Substituting the expression for $\mathrm{PID}(c)$ yields
\[
\frac{\gamma_0 \alpha_B(c)}
     {\gamma_0 \alpha_B(c) + \gamma_1 \beta_C(c)}
< 1-c.
\]
Since all quantities are positive and $c \in (0,1)$, this inequality is equivalent to
\begin{align*}
\gamma_0 \alpha_B(c)
&<
(1-c)\bigl[\gamma_0 \alpha_B(c) + \gamma_1 \beta_C(c)\bigr] \\
\gamma_0 \alpha_B(c)
&<
(1-c)\gamma_0 \alpha_B(c) + (1-c)\gamma_1 \beta_C(c) \\
\gamma_0 \alpha_B(c) - (1-c)\gamma_0 \alpha_B(c)
&<
(1-c)\gamma_1 \beta_C(c) \\
c\,\gamma_0 \alpha_B(c)
&<
(1-c)\gamma_1 \beta_C(c).
\end{align*}
Rearranging gives
\[
\beta_C(c)
>
\frac{\gamma_0}{\gamma_1}\cdot\frac{c}{1-c}\cdot \alpha_B(c),
\]
or equivalently,
\[
\frac{\beta_C(c)}{\alpha_B(c)}
>
\frac{\gamma_0}{\gamma_1}\cdot\frac{c}{1-c}.
\]
Hence, under large-sample settings, the comparison between $\mathrm{PID}(c)$ and the frequentist Type~I error rate is approximately governed by
\[
\mathrm{PID}(c) < \alpha(c)
\quad \Longleftrightarrow \quad
\frac{\beta_C(c)}{\alpha_B(c)}
>
\frac{\gamma_0}{\gamma_1}\cdot\frac{c}{1-c}.
\]
\end{proof}

\subsection{Proof of Proposition~\ref{prop:PID_T1E_bounds}}

\begin{proof}
Recall that
\[
\mathrm{PID}(c) = \Pr(\theta \le \delta \mid \text{success}),
\]
and therefore, by the definition of probability, $0 \le \mathrm{PID}(c) \le 1$. Under the large-sample setting stated in Proposition~\ref{prop:PID_T1E_bounds}, the frequentist Type I error rate satisfies $\alpha(c) \rightarrow 1-c$.
To establish the upper bound, we note $\mathrm{PID}(c) \le 1$:
\[
\mathrm{PID}(c) - \alpha(c) \le 1 - (1-c) = c.
\]
To establish the lower bound, we note $\mathrm{PID}(c) \ge 0$:
\[
\mathrm{PID}(c) - \alpha(c) \ge 0 - (1-c) = -(1-c).
\]
Combining the two inequalities yields
\[
-(1-c) \le \mathrm{PID}(c) - \alpha(c) \le c.
\]
Now consider the limits with respect to the design prior probability of effectiveness, denoted by $\gamma_1 = \Pr_{\mathrm{d}}(\theta > \delta)$ and $\gamma_0 = 1-\gamma_1$. The PID can be written as
\[
\mathrm{PID}(c)
=
\frac{\gamma_0 \alpha_B(c)}
     {\gamma_0 \alpha_B(c) + \gamma_1 \beta_C(c)},
\]
where $\alpha_B(c)$ is the Bayesian Type I error rate and $\beta_C(c)$ is the Bayesian conditional power. Under standard regularity conditions, assume $\alpha_B(c) > 0$ and $\beta_C(c) > 0$ for all $c \in (0,1)$.
First, as $\gamma_1 \to 1$ (so that $\gamma_0 \to 0$), the numerator converges to $0$ while the denominator converges to $\beta_C(c) > 0$. Hence,
\[
\lim_{\gamma_1 \to 1} \mathrm{PID}(c) = 0.
\]
Therefore,
\[
\lim_{\gamma_1 \to 1} \bigl[\mathrm{PID}(c) - \alpha(c)\bigr]
=
0 - (1-c)
=
-(1-c).
\]
Second, as $\gamma_1 \to 0$ (so that $\gamma_0 \to 1$), the term $\gamma_1 \beta_C(c)$ vanishes, and thus
\[
\lim_{\gamma_1 \to 0} \mathrm{PID}(c)
=
\frac{1\cdot \alpha_B(c)}{1\cdot \alpha_B(c) + 0}
=
1.
\]
Therefore,
\[
\lim_{\gamma_1 \to 0} \bigl[\mathrm{PID}(c) - \alpha(c)\bigr]
=
1 - (1-c)
=
c.
\]
\end{proof}

\newpage
\section{Analytical derivation of operating characteristics for different types of endpoints}

In this section, we first specify the commonly used estimands for each endpoint. Then, we derive analytical expressions for the Bayesian operating characteristic metrics for continuous and binary endpoints. For other non-normal endpoint, such as time-to-event outcomes, we show that the associated test statistics can be embedded within the same continuous framework, as their sampling distributions can be well approximated by a normal distribution via the Central Limit Theorem.

\subsection{Continuous endpoint} \label{sec:continuous_endpoint}

\subsubsection{Analytical derivation of operating characteristics}

For continuous endpoints, we consider both single-arm and two-arm trial settings under a unified framework. Let $\theta$ denote the estimand of interest. In a single-arm trial, we assume $Y_{T,i} \stackrel{\text{i.i.d.}}{\sim} \mathcal{N}(\theta_T,\sigma_T^2)$ for $i=1,\dots,n_T$, where $\theta=\theta_T$ denotes the mean outcome in the treatment arm. In a two-arm trial, we define the estimand as $\theta=\theta_T-\theta_C$, where $\theta_T$ and $\theta_C$ denote the mean outcomes in the treatment and control arms, respectively. Under the normality assumption, the estimator $\bar{Y}$ satisfies 
\begin{equation} \label{eqn:Y_bar}
    \bar{Y} \sim \mathcal{N}(\theta, v^2)
\end{equation} 
where in a single-arm trial $\bar{Y}=\bar{Y}_T$ and $v^2=\sigma_T^2/n_T$, while in a two-arm randomized trial $\bar{Y}=\bar{Y}_T-\bar{Y}_C$ and $v^2=\sigma_T^2/n_T+\sigma_C^2/n_C$.

Assuming the analysis prior $\pi_{\mathrm{a}}(\theta) = N(\theta_{\mathrm{a}}, \sigma^2_{\mathrm{a}})$, by standard normal-normal conjugacy, the posterior distribution is exactly normal:
\begin{equation*}
\theta \mid \mathcal{D} \sim \mathcal{N}(\mu_{\mathrm{post}}, \sigma_{\mathrm{post}}^2),
\end{equation*}
where the posterior variance and mean are respectively given by:
\begin{equation*}
\sigma_{\mathrm{post}}^2 = \left( \frac{1}{\sigma_{\mathrm{a}}^2} + \frac{1}{v^2} \right)^{-1}, \qquad \mu_{\mathrm{post}} = w \bar{Y} + (1-w) \theta_{\mathrm{a}},
\end{equation*}
with $w = \sigma_{\mathrm{post}}^2 / v^2$ representing the precision weight afforded to the observed data. (Note that as the analysis prior becomes non-informative, $\sigma_{\mathrm{a}}^2 \to \infty$, we recover $w \to 1$, $\mu_{\mathrm{post}} \to \bar{Y}$, and $\sigma_{\mathrm{post}}^2 \to v^2$).

Conditional on a prespecified clinical margin $\delta$ and a posterior probability threshold $c\in(0,1)$, the Bayesian success criterion in \eqref{eqn:decision_rule} evaluates to:
\begin{equation*}
\Pr(\theta>\delta \mid \text{data})
= \Phi\!\left(\frac{\mu_{\mathrm{post}} - \delta}{\sigma_{\mathrm{post}}}\right) > c,
\end{equation*}
where $\Phi(\cdot)$ denotes the standard normal cumulative distribution function. By isolating $\bar{Y}$, the success indicator \eqref{eqn:decision_indicator} reduces to a deterministic boundary on the sample mean:
\begin{equation*}
\mathcal{S}(\text{data};c)
= \mathbb{I}\!\left\{\bar{Y} > y_c \right\}, \quad y_c = \frac{1}{w} \left[ \delta - (1-w)\theta_{\mathrm{a}} + \Phi^{-1}(c)\sigma_{\mathrm{post}} \right].
\end{equation*}

To evaluate operating characteristics, we specify the design prior as
$\pi_{\mathrm{d}}(\theta)=\mathcal{N}(\theta_{\mathrm{d}},\sigma_{\mathrm{d}}^2)$. This prior represents pre-trial beliefs about the true data-generating parameter $\theta$ and is distinct from the analysis prior used for inference for observed data alone. Under the joint model defined by the design prior and the sampling distribution
$\bar{Y}\mid\theta \sim \mathcal{N}(\theta,v^2)$,
the pair $(\theta,\bar{Y})$ follows a bivariate normal distribution:
\[
\begin{pmatrix}
\theta \\
\bar{Y}
\end{pmatrix}
\sim
\mathcal{N}_2\!\left(
\begin{pmatrix}
\theta_{\mathrm{d}} \\
\theta_{\mathrm{d}}
\end{pmatrix},
\begin{pmatrix}
\sigma_{\mathrm{d}}^2 & \sigma_{\mathrm{d}}^2 \\
\sigma_{\mathrm{d}}^2 & \sigma_{\mathrm{d}}^2 + v^2
\end{pmatrix}
\right).
\]
We define the standardized random variables
\[
U=\frac{\theta-\theta_{\mathrm{d}}}{\sigma_{\mathrm{d}}}, 
\qquad 
V=\frac{\bar{Y} -\theta_{\mathrm{d}}}{\sqrt{\sigma_{\mathrm{d}}^2+ v^2}},
\]
which are jointly standard normal with correlation
\[
\rho=\mathrm{Corr}(\theta,\bar{Y})
=\frac{\sigma_{\mathrm{d}}}{\sqrt{\sigma_{\mathrm{d}}^2+v^2}}.
\]
Let the standardized boundaries be defined as
\[
a=\frac{\delta-\theta_{\mathrm{d}}}{\sigma_{\mathrm{d}}},
\qquad
b=\frac{y_c  -\theta_{\mathrm{d}}}{\sqrt{\sigma_{\mathrm{d}}^2+v^2}}.
\]
Then the effectiveness event $\mathcal{E} = \mathbb{I}\{ \theta > \delta \}$ can be equivalently written as $\mathbb{I}\{U>a\}$, while the success indicator $\mathcal{S}(\text{data}; c) = \mathbb{I}\!\left\{\bar{Y} > y_c \right\}$ is also equivalent to $\mathbb{I}\{V>b\}.$ Next, we derived the analytical solution for each operating characteristic metric. 

%========================================================
\begin{itemize}
%========================================================

%--------------------------------------------------------
\item \textbf{Bayesian power.}
%--------------------------------------------------------
\begin{equation*}
\boxed{
\Pr(\text{success})
= \Pr(V>b) = 1-\Phi(b).
}
\label{eqn:Bayes_power}
\end{equation*}
%--------------------------------------------------------
\item \textbf{Bayesian conditional power.}
%--------------------------------------------------------
\begin{equation*}
\boxed{
\Pr(\text{success}\mid \text{effective})
=
\frac{\Phi_2(a,b;\rho)}{\Phi(a)}.
}
\end{equation*}

%--------------------------------------------------------
\item \textbf{Bayesian Type I error rate.}
%--------------------------------------------------------
\begin{equation*}
\boxed{
\Pr(\text{success}\mid \text{ineffective}) =  \Pr(V>b\mid U\le a)
= \frac{\Phi(a)-\Phi_2(a,b;\rho)}{\Phi(a)}.
}
\label{eqn:BFPR}
\end{equation*}

%--------------------------------------------------------
\item \textbf{Probability of incorrect decision.}
%--------------------------------------------------------
\begin{equation*}
\boxed{
\Pr(\text{ineffective} \mid \text{success}) = \Pr(U\le a \mid \ V>b) = \frac{\Phi(a)-\Phi_2(a,b;\rho)}{1-\Phi(b)}.
}
\label{eqn:PID}
\end{equation*}

%--------------------------------------------------------
\item \textbf{False omission rate.}
%--------------------------------------------------------
\begin{equation*}
\boxed{
\Pr(\text{effective}\mid \text{failure}) = \Pr(U>a \mid V\le b)
= \frac{\Phi(b)-\Phi_2(a,b;\rho)}{\Phi(b)}.
}
\label{eqn:FOR}
\end{equation*}

%--------------------------------------------------------
\item \textbf{Frequentist Type I error rate.}
%--------------------------------------------------------
\begin{equation*}
\boxed{
\alpha(c)
= \Pr(\text{success}\mid \theta=\delta)= 1 - \Phi(\frac{y_c - \delta}{v}).
}
\label{eq:cox_type1_closed}
\end{equation*}

%========================================================
\end{itemize}
%========================================================

Although this analytical framework is developed for continuous endpoints, it can be readily extended to other non-normal endpoints, such as time-to-event outcomes, since their associated test statistics are approximately normally distributed under the Central Limit Theorem when the sample size is moderately large. Once the variance term in Equation~\eqref{eqn:Y_bar} is specified, the analytical solutions for the operating characteristic metrics apply equivalently. This condition is typically satisfied in confirmatory trials, which are the primary focus of this work. In the Section~\ref{sec:TTE}, we illustrate how time-to-event endpoints can be incorporated within the same framework through their corresponding large-sample normal approximations.

\subsubsection{Additional results for operating characteristics}

In this subsection, we evaluate how the dispersion of the design prior impacts the operating characteristics for continuous endpoints. Specifically, we vary the prior standard deviation, $\sigma_{\mathrm{d}} \in \{0.10, 0.15, 0.20\}$, while maintaining a neutral design prior centered at a mean of $E_{\mathrm{d}}(\theta) = 0$. The corresponding operating characteristics are presented in Figure~\ref{fig:oc_normal_sigmad}. Varying $\sigma_{\mathrm{d}}$ primarily affects the strength of the influence from the design prior on the operating characteristics. When the design prior is more concentrated ($\sigma_{\mathrm{d}}=0.10$), a larger portion of the prior probability mass is concentrated near the clinical margin, making it more difficult to distinguish effective from ineffective treatments. As a result, error-related metrics such as $\mathrm{PID}(c)$ and $\mathrm{FOR}(c)$ increase, while Bayesian conditional power and Bayesian power decrease across the range of posterior cutoffs. As the prior becomes more diffuse ($\sigma_{\mathrm{d}}=0.20$), less probability mass is concentrated near the clinical margin, so trials are more likely to generate data corresponding to clearly effective or clearly ineffective effects. Consequently, the operating characteristics move closer to those driven primarily by the sampling model, with modest reductions in error rates and slight improvements in power. Overall, these results illustrate that the design-prior standard deviation $\sigma_{\mathrm{d}}$ governs the degree of prior influence: smaller values concentrate probability mass near the clinical margin and inflate error rates while reducing power, whereas larger values yield behavior closer to that determined by the sampling model.

\begin{figure}[htbp]
\centering
\includegraphics[width=1\linewidth]{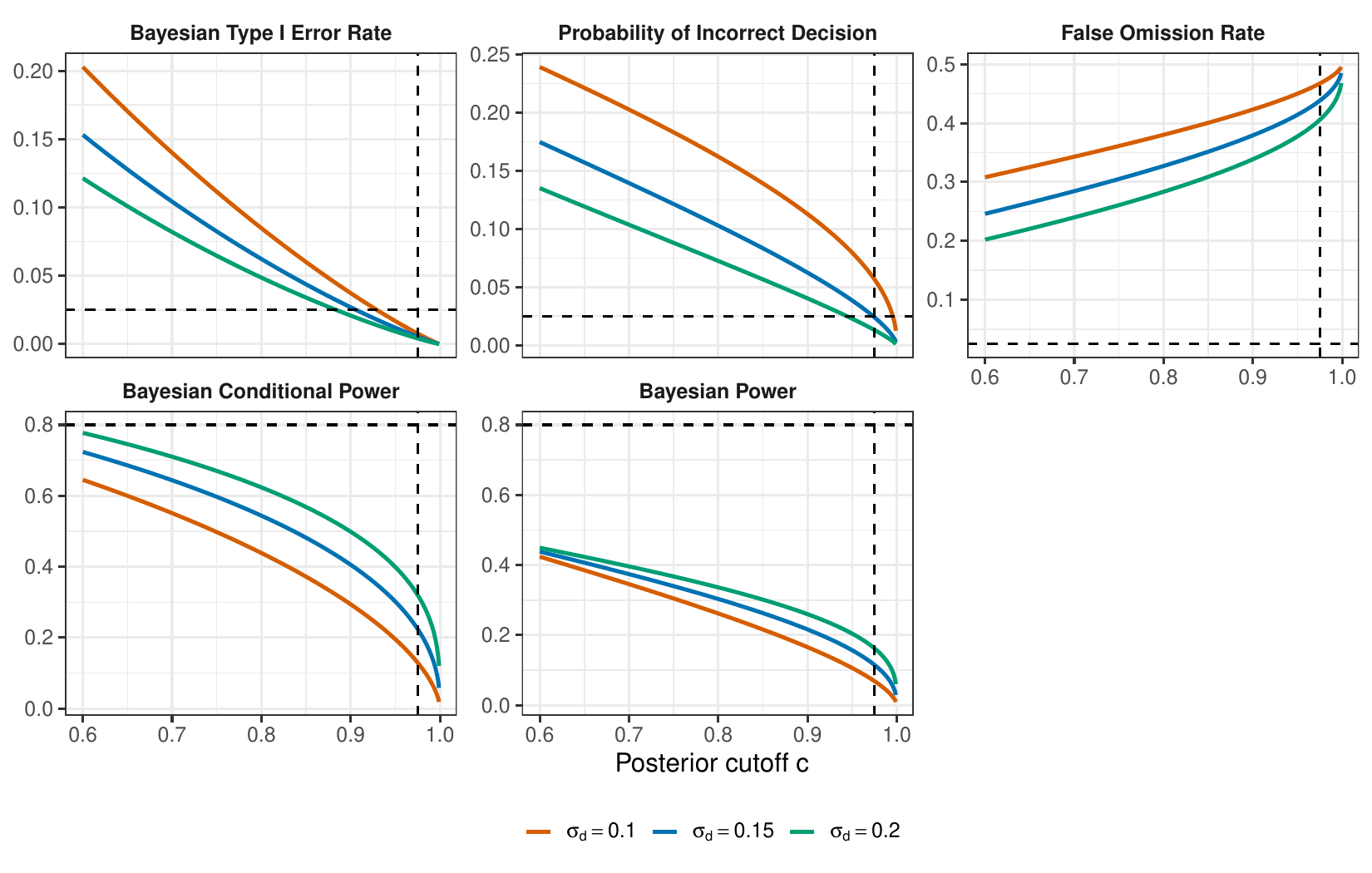}
\caption{Operating characteristics for single-arm continuous endpoints under a neutral design prior ($N(0, \sigma^2_{\mathrm{d}})$), evaluated for design-prior standard deviations $\sigma_{\mathrm{d}} = 0.10$, $0.15$, and $0.20$ (different colors). The vertical dashed line marks $c = 0.975$. Horizontal dashed lines indicate reference levels of $0.025$ (error-rate panels) and $0.80$ (power panels). The trial design assumes a clinical margin of $\delta = 0$, a sample size of $n_T = 74$, variance $\sigma = 1$, and a non-informative analysis prior ($\pi_{\mathrm{a},T}(\theta_T) = N(0,10^6)$).}
\label{fig:oc_normal_sigmad}
\end{figure}

\subsection{Binary endpoint}
\label{sec:Binary_exact}
\subsubsection{Single-arm trial design}
In this subsection, we derive analytical operating characteristic formulas for a single-arm binary endpoint. Let $Y_{T,i}\in\{0,1\}$ denote the binary outcome for patient $i=1,\dots,n_T$, and assume $Y_{T,i}\mid \theta \stackrel{\text{i.i.d.}}{\sim} \mathrm{Bernoulli}(\theta)$, where $\theta \in(0,1)$ denotes the response probability in the treatment arm, which also serves as the estimand of interest. Let $X=\sum_{i=1}^{n_T} Y_{T,i}$ denote the total number of responders, so that $X\mid \theta \sim \mathrm{Binomial}(n_T,\theta)$.

Suppose the design prior is $\pi_{\mathrm{d}}(\theta) = \mathrm{Beta}(a_{\mathrm{d}},b_{\mathrm{d}})$, while the analysis prior used for posterior decision-making is $\pi_{\mathrm{a}}(\theta) = \mathrm{Beta}(a_{\mathrm{a}},b_{\mathrm{a}})$. Under the Beta--Binomial conjugate model, the posterior under the analysis prior is
\[
\theta\mid X=x \sim \mathrm{Beta}(a_{\mathrm{a}}+x,\; b_{\mathrm{a}}+n_T-x).
\]
Hence, the posterior probability of effectiveness is
\[
\Pr_{\mathrm{a}}(\theta>\delta \mid X=x)
=
1-I_{\delta}(a_{\mathrm{a}}+x,\; b_{\mathrm{a}}+n_T-x),
\]
where $I_{\delta}(a,b)$ denotes the regularized incomplete beta function. The Bayesian success rule declares success if
\[
\Pr_{\mathrm{a}}(\theta>\delta\mid X=x)>c.
\]
Equivalently, define the critical count
\[
x_c = \min\Bigl\{x\in\{0,\dots,n_T\}: 1-I_{\delta}(a_{\mathrm{a}}+x,\; b_{\mathrm{a}}+n_T-x)>c\Bigr\},
\]
with the convention $x_c=n_T+1$ if the inequality is never satisfied. Then the trial declares success if and only if $X\ge x_c$.

\paragraph{Predictive distribution} To evaluate the operating characteristics, integrating out $\theta$ under the design prior yields the Beta--Binomial predictive distribution
\[
\Pr(X=x)
=
\binom{n_T}{x}\frac{B(a_{\mathrm{d}}+x,\; b_{\mathrm{d}}+n_T-x)}{B(a_{\mathrm{d}},b_{\mathrm{d}})},
\qquad x=0,\dots,n_T,
\]
where $B(\cdot,\cdot)$ denotes the beta function. Let
\[
m_{\mathrm{d}}(x)
=
\binom{n_T}{x}\frac{B(a_{\mathrm{d}}+x,\; b_{\mathrm{d}}+n_T-x)}{B(a_{\mathrm{d}},b_{\mathrm{d}})}.
\]
The prior prevalence of effective treatments under the design prior is
\[
\gamma_1=\Pr_{\mathrm{d}}(\theta>\delta)=1-I_{\delta}(a_{\mathrm{d}},b_{\mathrm{d}}),
\qquad
\gamma_0=\Pr_{\mathrm{d}}(\theta\le\delta)=I_{\delta}(a_{\mathrm{d}},b_{\mathrm{d}}).
\]
We let
\[
q_{\mathrm{d}}(x)
:=
\Pr_{\mathrm{d}}(\theta>\delta\mid X=x)
=
1-I_{\delta}(a_{\mathrm{d}}+x,\; b_{\mathrm{d}}+n_T-x).
\]
Therefore, the joint probabilities of $(X=x,\mathcal{E})$ and $(X=x,\bar{\mathcal{E}})$ are
\[
\Pr(X=x,\mathcal{E}) = m_{\mathrm{d}}(x)\,q_{\mathrm{d}}(x),
\qquad
\Pr(X=x,\bar{\mathcal{E}})=m_{\mathrm{d}}(x)\,[1-q_{\mathrm{d}}(x)].
\]
Using the success rule $X\ge x_c$, the analytical expression for operating characteristics is obtained as follows.

\paragraph{Operating characteristics} The operating characteristics can be derived as follows:

%========================================================
\begin{itemize}
%========================================================

%--------------------------------------------------------
\item \textbf{Bayesian power.}
%--------------------------------------------------------
\[\boxed{
\mathrm{BP}(c)
=
\Pr(X\ge x_c)
=
\sum_{x=x_c}^{n_T} m_{\mathrm{d}}(x).}
\]
%--------------------------------------------------------
\item \textbf{Bayesian conditional power.}
%--------------------------------------------------------
\[\boxed{
\beta_C(c)
=
\Pr(X\ge x_c \mid \theta>\delta)
=
\frac{\sum_{x=x_c}^{n_T} m_{\mathrm{d}}(x)\,q_{\mathrm{d}}(x)}{\gamma_1}.}
\]
%--------------------------------------------------------
\item \textbf{Bayesian type I error rate.}
%--------------------------------------------------------
\[\boxed{
\alpha_B(c)
=
\Pr(X\ge x_c \mid \theta\le\delta)
=
\frac{\sum_{x=x_c}^{n_T} m_{\mathrm{d}}(x)\,[1-q_{\mathrm{d}}(x)]}{\gamma_0}.}
\]
%--------------------------------------------------------
\item \textbf{Probability of incorrect decision.}
%--------------------------------------------------------
\[\boxed{
\mathrm{PID}(c)
=
\Pr(\theta\le\delta \mid X\ge x_c)
=
\frac{\sum_{x=x_c}^{n_T} m_{\mathrm{d}}(x)\,[1-q_{\mathrm{d}}(x)]}{\mathrm{BP}(c)}.}
\]
%--------------------------------------------------------
\item \textbf{False omission rate.}
%--------------------------------------------------------
\[\boxed{
\mathrm{FOR}(c)
=
\Pr(\theta>\delta \mid X<x_c)
=
\frac{\sum_{x=0}^{x_c-1} m_{\mathrm{d}}(x)\,q_{\mathrm{d}}(x)}{1-\mathrm{BP}(c)}.}
\]
%--------------------------------------------------------
\item \textbf{Frequentist Type I error rate.}
%--------------------------------------------------------
\[\boxed{
\alpha(c)
=
\Pr(X\ge x_c \mid \theta=\delta)
=
\sum_{x=x_c}^{n_T} \binom{n_T}{x}\delta^x(1-\delta)^{n_T-x}.}
\]
\end{itemize}

Thus, all operating characteristics for the single-arm binary endpoint admit exact closed-form expressions under the Beta--Binomial framework, with the only numerical step being the determination of the critical count $x_c$ from the posterior cutoff $c$.

\begin{figure}[htbp]
    \centering
   \includegraphics[width=\linewidth]{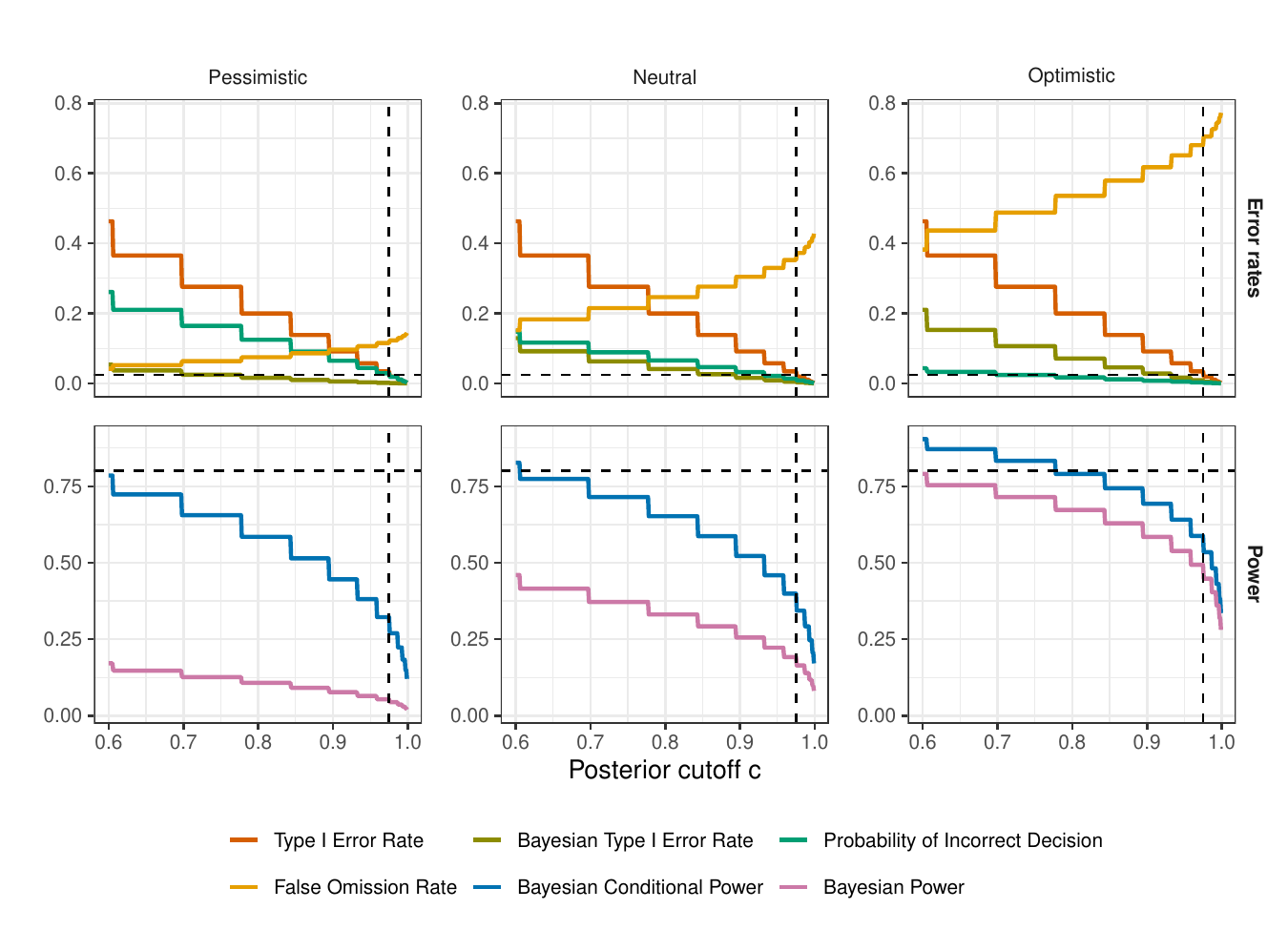}
    \caption{Operating characteristics for single-arm binary endpoints across posterior probability cutoffs $c$, evaluated under three design priors: pessimistic ($\pi_{\mathrm{d}, T}(\theta_T) = \text{Beta}(3,12)$), neutral ($\pi_{\mathrm{d}, T}(\theta_T) = \text{Beta}(6,14)$), and optimistic ($\pi_{\mathrm{d}, T}(\theta_T) = \text{Beta}(9,14)$), corresponding to $\gamma_1 = 0.257$, $0.443$, and $0.652$, respectively. The vertical dashed line marks $c = 0.975$, and the horizontal dashed lines denote 0.025 in the error-rate panel and 0.80 in the power panel. The trial design assumes a clinical margin of $\delta = 0.3$, a sample size of $n_T = 74$, and a non-informative analysis prior ($\pi_{\mathrm{a}, T}(\theta_T) = \text{Beta}(1,1)$).}
    \label{fig:oc_binary_sd_0.1}
\end{figure}
\vspace{0.5em}

Figure~\ref{fig:oc_binary_sd_0.1} presents the resulting operating characteristics evaluated under the exact Beta--Binomial model for the pessimistic, neutral, and optimistic design scenarios. As expected, increasing the posterior probability threshold $c$ imposes a stricter criterion for trial success, leading to monotonically decreasing trends across both error-related metrics (such as the Type I error rates and PID) and power-related metrics. The figure also highlights the strong influence of the design prior, with the optimistic scenario demonstrating substantially higher power, but correspondingly higher error rates, compared to the pessimistic setting at any given threshold.

Finally, the stepwise appearance of the curves reflects the discreteness of the Binomial data. Because the success rule depends on the integer count of responses $X$, the critical value $x_c$ changes only at discrete points as the posterior cutoff $c$ varies. Consequently, the operating characteristics change in a piecewise-constant manner, producing the distinctive staircase patterns observed in Figure~\ref{fig:oc_binary_sd_0.1}.

\subsubsection{Two-arm randomized trial design}

We now extend the operating characteristic formulas to a two-arm randomized trial with binary endpoints. Let $X_T$ and $X_C$ denote the number of events in the treatment and control arms, with sample sizes $n_T$ and $n_C$, respectively. We assume the outcomes are independent between arms such that $X_T \mid \theta_T \sim \mathrm{Binomial}(n_T, \theta_T)$ and $X_C \mid \theta_C \sim \mathrm{Binomial}(n_C, \theta_C)$. We assume the estimand of interest is the absolute risk reduction, i.e., $\theta = \theta_T - \theta_C$. Without loss of generality, the effectiveness region is defined as $\mathcal{E} = \{\theta > 0\}$.

We specify independent Beta distributions for the joint design prior, $\pi_{\mathrm{d}}(\theta_T, \theta_C) = \mathrm{Beta}(a_{\mathrm{d},T}, b_{\mathrm{d},T}) \times \mathrm{Beta}(a_{\mathrm{d},C}, b_{\mathrm{d},C})$, and the joint analysis prior, $\pi_{\mathrm{a}}(\theta_T, \theta_C) = \mathrm{Beta}(a_{\mathrm{a},T}, b_{\mathrm{a},T}) \times \mathrm{Beta}(a_{\mathrm{a},C}, b_{\mathrm{a},C})$. 

Under the analysis prior, the posterior probability of effectiveness given the observed data $(x_T, x_C)$ is denoted by $p_{\mathrm{a}}(x_T, x_C) = \Pr_{\mathrm{a}}(\theta > 0 \mid x_T, x_C)$, which can be evaluated via numerical integration of the independent Beta posteriors. The trial declares success if $p_{\mathrm{a}}(x_T, x_C) > c$. Let $\mathbb{I}_{\mathcal{S}}(x_T, x_C)$ denote the binary decision indicator:
\[
\mathbb{I}_{\mathcal{S}}(x_T, x_C) = 
\begin{cases} 
1 & \text{if } p_{\mathrm{a}}(x_T, x_C) > c, \\
0 & \text{otherwise.}
\end{cases}
\]

\paragraph{Predictive distribution}
Because both experimental arms are independent, the joint predictive distribution under the design prior is the product of the marginal Beta-Binomial distributions:
\[
m_{\mathrm{d}}(x_T, x_C) = 
\left[ \binom{n_T}{x_T}\frac{B(a_{\mathrm{d},T}+x_T,\; b_{\mathrm{d},T}+n_T-x_T)}{B(a_{\mathrm{d},T},b_{\mathrm{d},T})} \right]
\times
\left[ \binom{n_C}{x_C}\frac{B(a_{\mathrm{d},C}+x_C,\; b_{\mathrm{d},C}+n_C-x_C)}{B(a_{\mathrm{d},C},b_{\mathrm{d},C})} \right].
\]
The prior prevalence of effectiveness, $\gamma_1 = \Pr_{\mathrm{d}}(\mathcal{E})$, is obtained by integrating the joint design prior over the effective region $\theta_T > \theta_C$ given design prior. Similarly, the design posterior probability of effectiveness for any observed outcome pair is denoted by $q_{\mathrm{d}}(x_T, x_C) = \Pr_{\mathrm{d}}(\theta > 0 \mid x_T, x_C)$.

\paragraph{Operating characteristics} The operating characteristics are obtained as follows:
\begin{itemize}
%--------------------------------------------------------
\item \textbf{Bayesian power.}
%--------------------------------------------------------
\[\boxed{
\mathrm{BP}(c) = \sum_{x_T=0}^{n_T} \sum_{x_C=0}^{n_C} \mathbb{I}_{\mathcal{S}}(x_T, x_C) m_{\mathrm{d}}(x_T, x_C).}
\]
%--------------------------------------------------------
\item \textbf{Bayesian conditional power.}
%--------------------------------------------------------
\[\boxed{
\beta_C(c) = \frac{1}{\gamma_1} \sum_{x_T=0}^{n_T} \sum_{x_C=0}^{n_C} \mathbb{I}_{\mathcal{S}}(x_T, x_C) m_{\mathrm{d}}(x_T, x_C) q_{\mathrm{d}}(x_T, x_C).}
\]
%--------------------------------------------------------
\item \textbf{Bayesian Type I error rate.}
%--------------------------------------------------------
\[\boxed{
\alpha_B(c) = \frac{1}{\gamma_0} \sum_{x_T=0}^{n_T} \sum_{x_C=0}^{n_C} \mathbb{I}_{\mathcal{S}}(x_T, x_C) m_{\mathrm{d}}(x_T, x_C) [1 - q_{\mathrm{d}}(x_T, x_C)].}
\]
%--------------------------------------------------------
\item \textbf{Probability of incorrect decision.}
%--------------------------------------------------------
\[\boxed{
\mathrm{PID}(c) = \frac{\sum_{x_T=0}^{n_T} \sum_{x_C=0}^{n_C} \mathbb{I}_{\mathcal{S}}(x_T, x_C) m_{\mathrm{d}}(x_T, x_C) [1 - q_{\mathrm{d}}(x_T, x_C)]}{\mathrm{BP}(c)}.}
\]
%--------------------------------------------------------
\item \textbf{False omission rate.}
%--------------------------------------------------------
\[
\boxed{
\mathrm{FOR}(c) = \frac{\sum_{x_T=0}^{n_T} \sum_{x_C=0}^{n_C} [1 - \mathbb{I}_{\mathcal{S}}(x_T, x_C)] m_{\mathrm{d}}(x_T, x_C) q_{\mathrm{d}}(x_T, x_C)}{1 - \mathrm{BP}(c)}.}
\]
%--------------------------------------------------------
\item \textbf{Frequentist Type I error rate.}
%--------------------------------------------------------
\[
\boxed{
\alpha(c \mid \theta_0) = \sum_{x_T=0}^{n_T} \sum_{x_C=0}^{n_C} \mathbb{I}_{\mathcal{S}}(x_T, x_C) \binom{n_T}{x_T}\theta_0^{x_T}(1-\theta_0)^{n_T-x_T} \binom{n_C}{x_C}\theta_0^{x_C}(1-\theta_0)^{n_C-x_C},
}
\]
where $\theta_0$ is a fixed point null rate (e.g., the pooled event rate).
\end{itemize}

\begin{landscape}
\begin{table}[ht]
\centering
\caption{Summary of operating characteristic metrics and analytical solutions for binary endpoints under both single-arm and two-arm trial designs.}
\label{tbl:OC_binary_summary}
\renewcommand{\arraystretch}{2.2} 
% \resizebox{\textwidth}{!}{
\begin{tabular}{ccc}
\toprule
& \multicolumn{2}{c}{Analytical Solution} \\
\cmidrule(lr){2-3}
Metric & Single-Arm & Two-Arm \\
\midrule

\shortstack[c]{Bayesian power \\ $\beta_B(c)$} 
& $\displaystyle \sum_{x=x_c}^{n_T} m_{\mathrm{d}}(x)$ 
& $\displaystyle \sum_{x_T=0}^{n_T} \sum_{x_C=0}^{n_C} \mathbb{I}_{\mathcal{S}}(x_T, x_C) m_{\mathrm{d}}(x_T, x_C)$ \\

\shortstack[c]{Bayesian conditional\\  power $\beta_C(c)$} 
& $\displaystyle \frac{1}{\gamma_1} \sum_{x=x_c}^{n_T} m_{\mathrm{d}}(x)\,q_{\mathrm{d}}(x)$ 
& $\displaystyle \frac{1}{\gamma_1} \sum_{x_T=0}^{n_T} \sum_{x_C=0}^{n_C} \mathbb{I}_{\mathcal{S}}(x_T, x_C) m_{\mathrm{d}}(x_T, x_C) q_{\mathrm{d}}(x_T, x_C)$ \\

\shortstack[c]{Bayesian Type I \\  error rate $\alpha_B(c)$} 
& $\displaystyle \frac{1}{\gamma_0} \sum_{x=x_c}^{n_T} m_{\mathrm{d}}(x)\,[1-q_{\mathrm{d}}(x)]$ 
& $\displaystyle \frac{1}{\gamma_0} \sum_{x_T=0}^{n_T} \sum_{x_C=0}^{n_C} \mathbb{I}_{\mathcal{S}}(x_T, x_C) m_{\mathrm{d}}(x_T, x_C) [1 - q_{\mathrm{d}}(x_T, x_C)]$ \\

\shortstack[c]{Probability of incorrect \\  decision $\mathrm{PID}(c)$} 
& $\displaystyle \frac{\sum_{x=x_c}^{n_T} m_{\mathrm{d}}(x)\,[1-q_{\mathrm{d}}(x)]}{\mathrm{BP}(c)}$ 
& $\displaystyle \frac{\sum_{x_T=0}^{n_T} \sum_{x_C=0}^{n_C} \mathbb{I}_{\mathcal{S}}(x_T, x_C) m_{\mathrm{d}}(x_T, x_C) [1 - q_{\mathrm{d}}(x_T, x_C)]}{\mathrm{BP}(c)}$ \\

\shortstack[c]{False omission rate \\ $\mathrm{FOR}(c)$} 
& $\displaystyle \frac{\sum_{x=0}^{x_c-1} m_{\mathrm{d}}(x)\,q_{\mathrm{d}}(x)}{1-\mathrm{BP}(c)}$ 
& $\displaystyle \frac{\sum_{x_T=0}^{n_T} \sum_{x_C=0}^{n_C} [1 - \mathbb{I}_{\mathcal{S}}(x_T, x_C)] m_{\mathrm{d}}(x_T, x_C) q_{\mathrm{d}}(x_T, x_C)}{1 - \mathrm{BP}(c)}$ \\

\shortstack[c]{Frequentist Type I  \\ error rate $\alpha(c)$} 
& $\displaystyle \sum_{x=x_c}^{n_T} \binom{n_T}{x}\delta^x(1-\delta)^{n_T-x}$ 
& $\displaystyle \sum_{x_T=0}^{n_T} \sum_{x_C=0}^{n_C} \mathbb{I}_{\mathcal{S}}(x_T, x_C) \binom{n_T}{x_T}\theta_0^{x_T}(1-\theta_0)^{n_T-x_T} \binom{n_C}{x_C}\theta_0^{x_C}(1-\theta_0)^{n_C-x_C}$ \\

\bottomrule
\end{tabular}
% }
\begin{flushleft}
\small \textbf{Note}: $\mathcal{E}$ denotes the effectiveness event and $\mathcal{S} = 1$ indicates a declaration of trial success. For the single-arm design, $x_c$ denotes the critical count required for success. For the two-arm design, $\mathbb{I}_{\mathcal{S}}(x_T, x_C)$ denotes the binary success indicator evaluated at the observed event grid point $(x_T, x_C)$. The marginal predictive probability under the design prior, $m_{\mathrm{d}}$, and the posterior probability of effectiveness under the design prior, $q_{\mathrm{d}}$, are defined within their respective single-arm and two-arm sections. For the Frequentist Type I error rate in the two-arm setting, $\theta_0$ represents the assumed point null rate (e.g., the pooled event rate).
\end{flushleft}
\end{table}
\end{landscape}

\subsection{Time-to-event endpoint} \label{sec:TTE}

For time-to-event endpoints, we consider two-arm randomized controlled trials. Let $T$ denote the event time and $Z \in \{0,1\}$ the treatment indicator ($Z=1$ for treatment and $Z=0$ for control). We assume a Cox proportional hazards model $h(t \mid Z) = h_0(t)\exp(\theta Z)$, where $h_0(t)$ is an unspecified baseline hazard function and $\theta$ denotes the log-hazard ratio, serving as the estimand of interest. Because smaller values of $\theta$ correspond to greater treatment efficacy, the Bayesian decision rule evaluates whether the posterior probability of the treatment effect being strictly less than the clinical margin $\delta$ exceeds the threshold $c$; that is, trial success is declared if $\Pr(\theta < \delta \mid \mathcal{D}) > c$.

Let $\widehat{\theta}$ denote the partial-likelihood estimator of $\theta$. Under standard regularity conditions and for moderate-to-large sample sizes, the estimator satisfies the large-sample approximation
\[
\widehat{\theta} \sim \mathcal{N}(\theta, v^2),
\]
which follows from the asymptotic normality of the Cox partial-likelihood estimator.

In practice, the sampling variance $v^2$ depends primarily on the total number of observed events. For a two-arm trial with allocation proportion $r$ to the treatment arm and expected total number of events $D$, a common design-stage approximation is
\[
v^2 \approx \frac{1}{D\,r(1-r)}.
\]
During operating characteristic evaluation, we therefore specify a design-stage variance $v^2$ based on the planned number of events and allocation ratio. This allows the time-to-event setting to fit within the same normal framework presented in Equation~\eqref{eqn:Y_bar}, enabling closed-form evaluation of the corresponding posterior quantities and operating characteristics.

To illustrate this framework, we evaluate the operating characteristics for a time-to-event endpoint using the normal approximation to the log-hazard-ratio estimator. Specifically, the estimator of the treatment effect, $\widehat{\theta}=\log(\mathrm{HR})$, is assumed to follow an approximate normal distribution with variance estimated using the Schoenfeld approximation \citep{schoenfeld1983sample}, $\mathrm{Var}(\widehat{\theta})\approx 4/D$, corresponding to a balanced treatment allocation ($r=0.5$). The total number of observed events is fixed at $D=100$. The clinical margin is set to $\delta=0$ on the log-hazard-ratio scale. 

Importantly, unlike the previous binary case study where priors were assigned independently to the treatment and control arms, here we place the design and analysis priors directly on the relative treatment effect $\theta$. Design priors for $\theta$ are specified as Normal distributions with a common standard deviation $\sigma_{\mathrm{d}}=0.25$ and location parameters $\theta_{\mathrm{d}} \in \{-0.1, 0, 0.1\}$, representing optimistic, neutral, and pessimistic scenarios, respectively. Under each prior, the design-stage prevalence of effective trials is defined as $\gamma_1=\Pr_{\mathrm{d}}(\theta<0)$, corresponding to a true hazard ratio of less than one. These prevalence values decrease monotonically from the optimistic to the pessimistic setting. The resulting operating characteristics exhibit the expected monotone behavior in the cutoff $c$: as the posterior threshold becomes more stringent, error-rate metrics (with the exception of the FOR) decrease, while power-related metrics decline accordingly.

\begin{figure}[htbp]
\centering
\includegraphics[width=\linewidth]{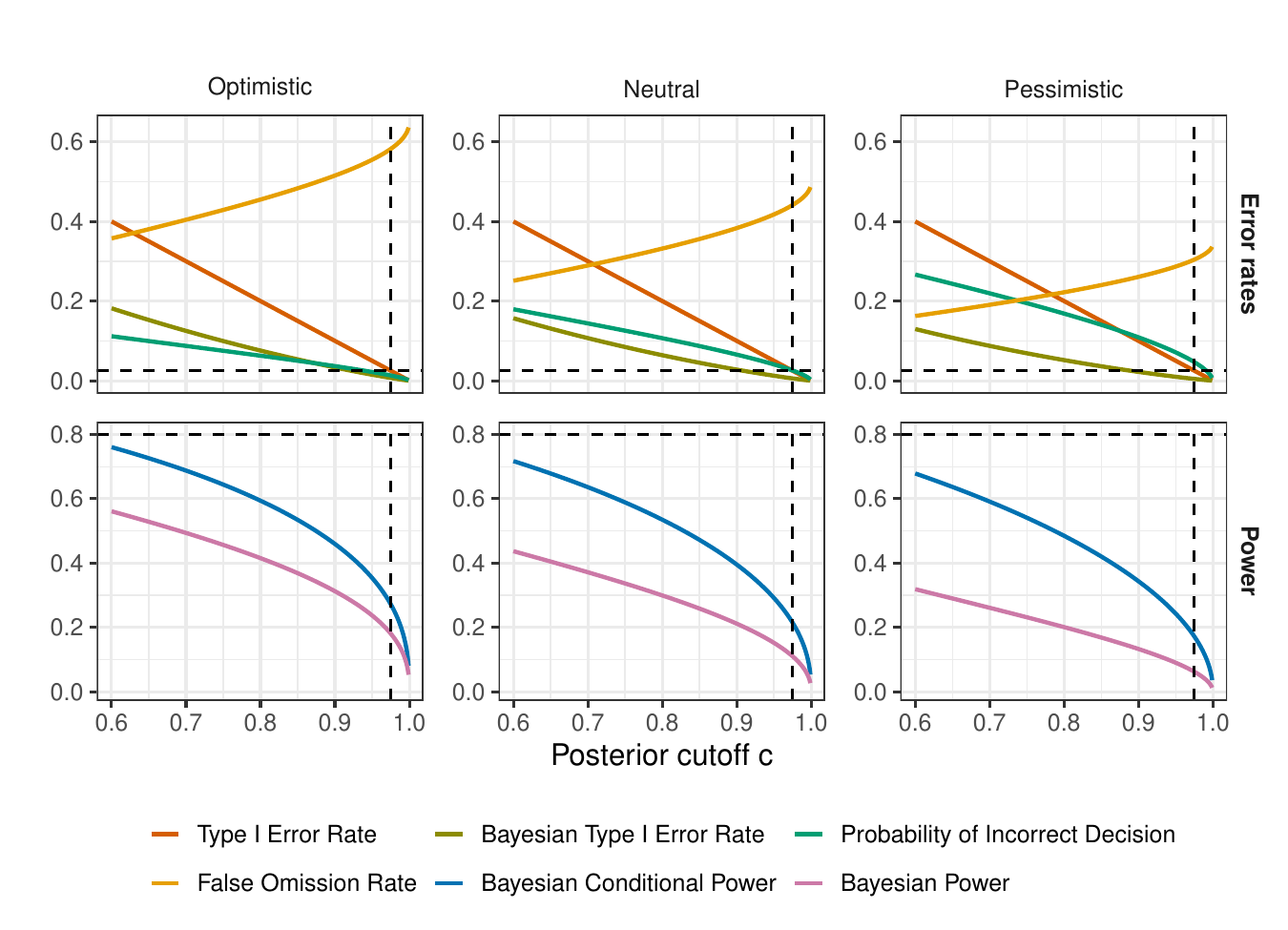}
\caption{Operating characteristics for two-arm time-to-event endpoints across posterior probability cutoffs $c$, evaluated under three design priors: optimistic ($\pi_{\mathrm{d}}(\theta) = N(-0.1, 0.25^2)$), neutral ($\pi_{\mathrm{d}}(\theta) = N(0, 0.25^2)$), and pessimistic ($\pi_{\mathrm{d}}(\theta) = N(0.1, 0.25^2)$), corresponding to $\gamma_1 = 0.655$, $0.500$, and $0.345$, respectively. The vertical dashed line marks $c = 0.975$, and the horizontal dashed lines denote 0.025 in the error-rate panel and 0.80 in the power panel. The trial design assumes a clinical margin of $\delta = 0$ and a non-informative analysis prior ($\pi_{\mathrm{a}}(\theta) = N(0, 10^6)$).}
\label{fig:oc_TTE}
\end{figure}

\end{document}